\documentclass{IEEEtran}
\usepackage{graphicx}
\usepackage[utf8]{inputenc}
\usepackage{subfigure}

\usepackage[colorlinks]{hyperref}

\usepackage{amsmath}
\usepackage{amssymb}
\usepackage{cite}

\usepackage[ruled]{algorithm2e}
\usepackage{algorithmic}
\usepackage{multirow}
\usepackage{bm}
\usepackage{mathtools}
\usepackage{multicol}
\usepackage{changepage}

\usepackage{footnote}
\graphicspath{ {./Figures/} }

\newtheorem{theorem}{Theorem}
\newtheorem{lemma}{Lemma}

\usepackage{array}
\makeatletter
\newcommand{\thickhline}{%
    \noalign {\ifnum 0=`}\fi \hrule height 1pt
    \futurelet \reserved@a \@xhline
}
\newcolumntype{"}{@{\hskip\tabcolsep\vrule width 1pt\hskip\tabcolsep}}
\makeatother

\begin{document}

\title{Performance Analysis of Fluid Antenna System Aided OTFS Satellite Communications
}

\author{Halvin~Yang,~\IEEEmembership{Member,~IEEE,}
            Mahsa Derakhshani,~\IEEEmembership{Senior Member,~IEEE,}
            Sangarapillai Lambotharan,~\IEEEmembership{Senior Member,~IEEE, }
            Lajos Hanzo,~\IEEEmembership{Life Fellow,~IEEE,}


}

\markboth{Submitted to IEEE Journal on Selected Areas in Communications, 2025}%
{6G Wireless System}

\maketitle

\begin{abstract}
Internet-of-Things (IoT) networks {typically} rely on satellite communications to provide coverage {in rural areas}. However, high-mobility satellite links introduce severe Doppler and delay spreads, which necessitate the use of orthogonal time frequency space (OTFS) modulation for reliable data transmission. {Furthermore, the space and energy constraints on IoT devices make the perfect use case for fluid antenna systems (FAS) due to their mechanical simplicity.} Hence, we propose a sophisticated FAS aided OTFS (FAS-OTFS) framework for satellite-based IoT networks. We derive analytical expressions for both the outage probability and ergodic capacity of FAS-OTFS under a general channel model, where the expressions derived are presented in integral form or as analytical bounds for efficient numerical evaluation. Additionally, we investigate a single-path fading scenario, where closed-form expressions are obtained. {Our} numerical results demonstrate significant performance gains in terms of both the outage probability and capacity compared to conventional OTFS systems, confirming the efficacy of FAS-OTFS in energy-constrained high-mobility environments. Our findings establish FAS-OTFS as a promising candidate for next-generation IoT communications over satellite links.
\end{abstract}

\begin{IEEEkeywords}
performance analysis, fluid antenna system, Orthogonal time frequency space modulation, LEO Satellite, Satellite communications, internet of things (IoT)
\end{IEEEkeywords}

\IEEEpeerreviewmaketitle

\section{Introduction}
\subsection{Background}
\IEEEPARstart{W}{ith} the proliferation of the Internet of Things (IoT) or other edge devices there is an increasing reliance on {their} global connectivity. However, the traditional terrestrial infrastructure is only able to cover 20-25\% of the world \cite{WorldBankOpenData}, hence may not be able to reach many IoT devices which are often located in rural environments. {As a remedy,} solutions like unmanned aerial vehicles (UAVs) or high altitude platforms (HAPs) have been proposed \cite{Alfattani2023}, but they {still} only provide limited coverage while having high operating expenses. {Hence they are} unsustainable {as} a permanent solution. A more promising area is the integration of terrestrial and satellite communication systems, especially low earth orbit (LEO) satellites, popularised {by} the recent success of Starlink. With {sufficient} LEO satellites in orbit, global coverage {may indeed} be provided \cite{Liu2021}. 

A significant challenge {in} non-terrestrial networks (NTNs) is the high mobility of satellites and their distance from the user \cite{Giordani2021}. This causes significant delays as well as Doppler shifts, which has to be addressed for reliable communications. Orthogonal time frequency space (OTFS) modulation \cite{Hadani2017} {constitutes} an effective solution to this problem. By operating in the delay-Doppler (DD) domain, OTFS {harmoniously harnesses} the underlying characteristics of the wireless channel, facilitating superior performance in terms of reliability and spectral efficiency. {Hence} OTFS has been shown to outperform traditional modulation schemes like OFDM \cite{Gaudio2022} in scenarios {having} severe Doppler shifts, making it an ideal candidate for satellite communications.

Both the computational power and energy are limited at many IoT or mobile devices, making the computationally intensive 2-D transform that OTFS requires challenging \cite{Xiao2022}. In order to increase both the computation and energy efficiency at the mobile device, we will demonstrate that the intrinsic integration of OTFS with fluid antenna systems (FAS) is beneficial. This is because FAS are capable of adaptively repositioning the antenna elements for maximizing the signal-to-noise ratio (SNR). {Hence they} have garnered attention as a means of mitigating the interference and improving the link reliability \cite{Wong2021}. A FAS is capable of achieving a performance comparable to multiple fixed antennas, while reducing both the energy consumption and space requirements, which are key advantages for resource-constrained mobile receivers. The energy efficiency of FAS also aligns with the increasing focus on sustainable wireless communication technologies \cite{Kumar2023}.

{Hence we intrinsically} integrate OTFS and FAS. Through the development of a suitable system model, its performance is analysed to demonstrate the benefits of this {beneficially amalgamated} system and to offer further insights.





\subsection{Related Work}

This transformative technique is capable of addressing the challenges of next generation (NG) wireless systems. We conclude with a comprehensive review of the relevant literature, focusing on the role of OTFS in high-mobility scenarios, on the computational and energy challenges of IoT and mobile devices, and on the potential of FAS as a disruptive technology. Additionally, the performance analysis of FAS is discussed, highlighting recent advances and their implications for integrated systems.

NG systems are envisioned to support massive connectivity, ultra-reliable low-latency communication (URLLC), and enhanced mobile broadband (eMBB) services \cite{Wang2023}. However, the traditional terrestrial infrastructure struggles to provide flawless connectivity in remote rural areas. While UAVs have been proposed as a solution, their long-term sustainability is limited by energy constraints and operational costs \cite{Shen2024}. Satellite communications (Satcom) have emerged as a viable {design} alternative, with an early channel model being proposed in \cite{Lutz1991}. This was then expanded upon by \cite{Abdi2003}, which included different environments, antennas and satellite elevations. However, in high-mobility scenarios, OFDM systems {falter in the face of} severe Doppler effects \cite{Kodheli2021}. OTFS, {which is} a DD domain modulation scheme first proposed in \cite{Hadani2017}, has been shown to outperform OFDM in such environments, making it a promising candidate for integrated satellite-terrestrial networks \cite{Gaudio2022}.

Previous analysis of OTFS shows excellent performance both in multiple access \cite{chen2025otfs} and massive MIMO \cite{mehrotra2024sparse}. Its integration with reflective intelligent surfaces (RIS) \cite{li2022hybridRIS}, NOMA \cite{ding2019otfsnoma} or even {in} combination with chirp waveforms {designed} for sensing \cite{zegrar2024otfschirp} show the versatility of using OTFS {along} with other emerging technologies. As a trade-off for its advantages, OTFS poses significant computational and energy {dissipation} challenges for IoT devices {(IoTDs)} \cite{xiao2022otfs}. The two-dimensional inverse symplectic finite Fourier transform (ISFFT) \cite{Mohammed2021} and  {the channel estimation} equalization process require substantial resources, which are often prohibitive for low-power IoTDs \cite{lakew2023intelligent}. Recent efforts have focused on reducing {the} OTFS complexity and energy consumption. For instance, \cite{surabhi2020lowcomplexity} proposed low-complexity equalization algorithms, while \cite{reddy2021nr_otfs} developed energy-efficient OTFS transceivers for IoT applications. However, scalability remains a critical issue, necessitating further research into lightweight OTFS implementations \cite{sui2023low,ge2021otfsreceiver}.

FAS have gained significant attention as a benefit of their ability to dynamically adjust the antenna configurations to optimize signal reception and transmission \cite{Wong2021,wong2023slowfluid}. This is achieved {at a modest} computational power, {since they do not need} channel estimation or beamforming \cite{new2024fluid}.  Recent studies have refined FAS performance analysis by incorporating {the} noise \cite{yang2023performance} into channel models. FAS has also been explored in the context of integrated data and energy transfer (IDET) \cite{lin2024performance}, {jointly} optimizing both data and energy transfer \cite{zhang2024joint,lin2024fluid}. These advances highlight the potential of FAS {in} enhancing {the} spectral vs. energy {efficiency} {of} NG networks. Apart from IDET, {its amalgam} with both index modulation \cite{yang2024position} as well as integrated sensing and communications (ISAC) \cite{wang2024fluid} shows the versatility of FAS. This adaptability as well as energy and computational efficiency makes FAS particularly {well} suited for {both} high-mobility and resource-constrained environments.

In summary, while significant progress has been made in the development of OTFS and FAS as standalone technologies, their integration remains a promising yet hitherto unexplored research direction. We address this gap by proposing an intrinsically {integrated FAS-OTFS} framework, with a focus on enhancing performance in high-mobility and resource-constrained scenarios.
%
%
%
%
\subsection{Motivation and Contributions}
This paper has made the following contributions:
\begin{itemize}
	\item A novel system model integrating FAS with OTFS in Satcoms is conceived. A general model considering many propagation paths is proposed to cover the most realistic use cases. By exploiting the unique features of FAS, like their port selection, the received SNR can be maximised. From this system model, the FAS-OTFS channel is established.
	\item Due to the complexity of the channel as well as the presence of correlation {imposed by the} FAS, there is no known distribution that describes the channel statistics. {This paper gives the formulation for the} statistical characterisation of the channel as a Gamma approximation and its features {are} derived. 
	\item The performance metrics of outage probability and ergodic capacity are derived from the SNR. Due to the presence of correlation it is {challenging} to provide a closed-form solution for the general model. For further insight, {we provide bounds} to obtain closed-form solutions.
	\item A specific scenario is considered, where there is only a single line-of-sight (LoS) path from the satellite to user. Although less general than the previous model, it is still a widely used assumption \cite{Lyras2017} due to the {specific} nature of Satcoms. By considering this special case, {this paper derives the approximations for} the closed-form expressions {of} both the outage probability and ergodic capacity using Gaussian quadratures.
\end{itemize}

The rest of this paper is organized as follows. Section \ref{sec:system_model} introduces the FAS-OTFS model followed by the statistical characterisation of the channel in Section \ref{sec:characterisation}. {The} outage probability and ergodic capacity {analysis of our}  general model is {carried out} in Section \ref{sec:general}, {and} a closed form expression is provided {for these metrics}.
%
\section{System Model}\label{sec:system_model}
{The aim of this section is to construct a correlated FAS-OTFS channel, which has not been {disseminated in the open} literature. In order to achieve this, first a standard OTFS channel {is defined} in the DD domain along with its parameters. Then, features of the FAS, most significantly {its} correlation and ports, are characterised and added to this channel model in order to obtain the FAS-OTFS channel that will be used {in} the rest of this paper. }
\subsection{Propagation Scenario}
First, the propagation scenario is discussed where the physical settings of the FAS-OTFS communication model {are} defined.

\begin{itemize}[]
	\item \textbf{LEO Satellite scenario:} A LEO satcom model is considered, where a satellite {having} a single antenna transmits to a mobile user, which has a single FA. The FA {user} is located in an urban environment with several natural obstacles like mountains and trees, which generate {multiple} paths, while maintaining a LoS path.
	\item \textbf{Bandwidth and carrier frequency:} The most standard choice for LEO modelling is the Ku-band \cite{panagopoulos2004satellite}. A centre frequency $f_c \approx 12$~GHz or 14~GHz {will be used} for link simulations.
	\item \textbf{Beam Pattern:} A broader coverage beam was considered to ensure signals can arrive at the receiver from various angles after {reflection from} different structures \cite{xia2019beam}.
	\item \textbf{Transmit power:} 	A higher transmit power is considered to make weaker paths more detectable at the receiver \cite{talgat2024maximizing}. A moderate-to-high transmit power alongside a broad beam can ensure the presence of reflected paths and that these paths remain above the noise floor. 
\end{itemize}



\noindent 

\noindent 

\noindent 
 
%
\subsection{Uncorrelated DD OTFS Satcom channel}
%
%
%
%
%
\subsubsection{LoS Path}

\begin{itemize}
	\item The delay of the direct path is $
	\tau_{\mathrm{LOS}} \approx \frac{d_{\mathrm{sat-user}}}{c}$, where $d_{\mathrm{sat-user}}$ is the {distance} from satellite to user.
	\item The Doppler of this path is
	$\nu_{\mathrm{LOS}} \approx \frac{f_c}{c}\cdot \mathbf{v}\cdot \hat{\mathbf{r}}$, where $\mathbf{v}$ is {the} satellite velocity, $\hat{\mathbf{r}}$ is the unit vector from {the} user to satellite, and $f_c$ is the carrier frequency.
	\item The gain of this path is  
	$\tilde{\alpha}_{\mathrm{LOS}} \sim \sqrt{\beta_{\mathrm{LOS}}} e^{j\phi_{\mathrm{LOS}}}$, where $\beta_{\mathrm{LOS}}$ is {the attenuation of} the LOS path (affected by distance, free-space path loss, shadowing, satellite antenna gain, etc.), and $\phi_{\mathrm{LOS}}$ is {its} phase.
\end{itemize}

\subsubsection{Multipath Clusters}
For LEO Satcoms, a Rician fading channel model is used due to the strong LoS presence.  For a total of $P$ clusters, each of which represents a combined group of paths as a result of reflections from objects in the environment, the various parameters for the $p$-th cluster can be defined as:

\begin{itemize}
	\item The delay is $\tau_p \approx \tau_{\mathrm{LOS}} + \delta \tau_p$, where $ \delta \tau_p \sim \text{Exp}(\lambda_\tau) $ reflects excess delay due to multipath propagation.
	\item The Doppler is $\nu_p \approx \nu_{\mathrm{LOS}} + \delta \nu_p$, where $\delta \nu_p \sim  \mathcal{U}(-\nu_{\max}, \nu_{\max}) $ models the spread of Doppler offsets due to relative motion.
	\item The gain of the $p$-th path is $\tilde{\alpha}_p \sim \sqrt{\beta_p} e^{j\phi_p}$, where $\phi_p \sim \mathcal{U}(0, 2\pi) $ models the random phase of the complex gain in multipath environments and $\beta_p$ is the power attenuation of the $p$-th cluster.
\end{itemize}
	\subsubsection{Delay-Doppler Representation in OTFS}
	
	In OTFS, the channel in the DD domain can now be expressed as:%
	\begin{align}\label{eqn:general_OTFS}
	h(\tau,\nu) =& \tilde{\alpha}_{\mathrm{LOS}}\,\delta(\tau - \tau_{\mathrm{LOS}})\,\delta(\nu - \nu_{\mathrm{LOS}}) \\ \nonumber
	\;&+\;\sum_{p=2}^{P}\tilde{\alpha}_{m}\,\delta(\tau - \tau_{m})\,\delta(\nu - \nu_{m}),
	\end{align}
	\noindent where $p=1$ is the dominant LoS path.
\subsection{Correlated FAS-OTFS Channel }
In order to incorporate FAS into the channel model, some key characteristics of FAS {have} to be defined. A single FA located at the receiver has a size of $W\lambda$, {constituted by} a total of $N$ {equi-}spaced ports. Therefore, the location of the $k$-th port on the FA would be {at the position} $\frac{k - 1}{N - 1}W$, and the {signal} correlation matrix between the ports $k,l$ is \cite{Wong2021}:%
\begin{align}\label{eqn:correlation}
[\mathbf{R}]_{k,l} = J_{0} \Bigl( 2\pi\frac{|k - l|}{N - 1} W
\Bigr), \quad 1 \le k,l \le N,
\end{align}

\noindent where $[\mathbf{R}]_{k,l}$ is the correlation between the $k$-th and $l$-th ports and $J_{0}$ is the zeroth order Bessel function {of the first kind}. Having defined the port characteristics, the overall FAS-OTFS channel can be obtained from (\ref{eqn:general_OTFS}) {as}:%
\begin{align}\label{eqn:FAS-OTFS_channel}
h_{k}(\tau,\nu) = \sum_{p=1}^{P} \tilde{\alpha}_{p} c_{p,k} \delta(\tau - \tau_{p}) \delta(\nu - \nu_{p}),
\end{align}

\noindent where $c_{p,k}$ is the random complex gain of path $p$ at port $k$. {Note that this model is more generalised than \ref{eqn:general_OTFS} and the LoS path is just considered one of the many Rician clusters.} The correlation is represented in this complex gain, with: %
\begin{align}\label{eqn:correlation_matrix}
\mathbb{E} \bigl[ \mathbf{c}_{p}\,\mathbf{c}_{p}^{H} \bigr] = \mathbf{R},
\quad
\mathbf{c}_{p} = \bigl[c_{p,1},\,c_{p,2},\dots,\,c_{p,N} \bigr]^{T}.
\end{align}
%
%
%
\section{Statistical Characterisation of FAS-OTFS Channel} \label{sec:characterisation}
With the FAS-OTFS channel defined, {its} statistical characterisation can be obtained. First, the distribution of the random variable $c_{p,k}$ {has} to be found. Then, the overall distribution of the summation of $P$ clusters can be found. However, since this summation is {un}likely {to obey any} known distribution, a suitable approximation is used to express $h_k$. {Once} the expression for the channel is obtained, the distribution of the SNR can be derived.

Since a Rician cluster is assumed, the channel gain $c_{p,k}$ representing the complex gain of the $p$-th cluster at the $k$-th port, can be separated into a deterministic and {a} diffuse component. For the $p$-th cluster at the $k$-th port {we have} $c_{p,k} = \mu_{p,k} + \tilde{c}_{p,k}$, where $\mu_{p,k}$ is a deterministic complex constant (the specular component) at port $k$ and $\tilde{c}_{p,k} \sim \mathcal{CN}\!\Bigl(0, [\mathbf{R}]_{k,k}\Bigr)$ is the diffuse random variable from the common correlation matrix $\mathbf{R}$. Although $[\mathbf{R}]_{k,k}=1$, {the} correlation can be {formulated} between different ports as $\mathbb{E}\bigl[\tilde{c}_{p,k}\tilde{c}_{p,l}^{*}\bigr] = \bigl[\mathbf{R}\bigr]_{k,l}$.

Since both the FAS feature of port selection as well as the {SNR} calculations and outage probability all use the channel magnitude, the distribution of $|h_k|^2$ is derived instead of $h_k$. Due to the complex nature of the squared magnitude of a sum of Rician random variables, no known distribution can directly characterise $|h_k|^2$. Therefore, an approximation is required for mathematical tractability. Although a Gaussian approximation via the central limit theorem (CLT) is generally preferred due to the ease of later analysis, Figure \ref{fig:approximation_comparison} shows that for small $P$ the Gaussian approximation is very inaccurate. Since in a LEO satellite FAS-OTFS channel {we only have a few} scatterers, in general $P \leq 5$, {the} Gamma approximation is more suitable.
\begin{figure}[!t]
	\centering
	\includegraphics[width=3in]{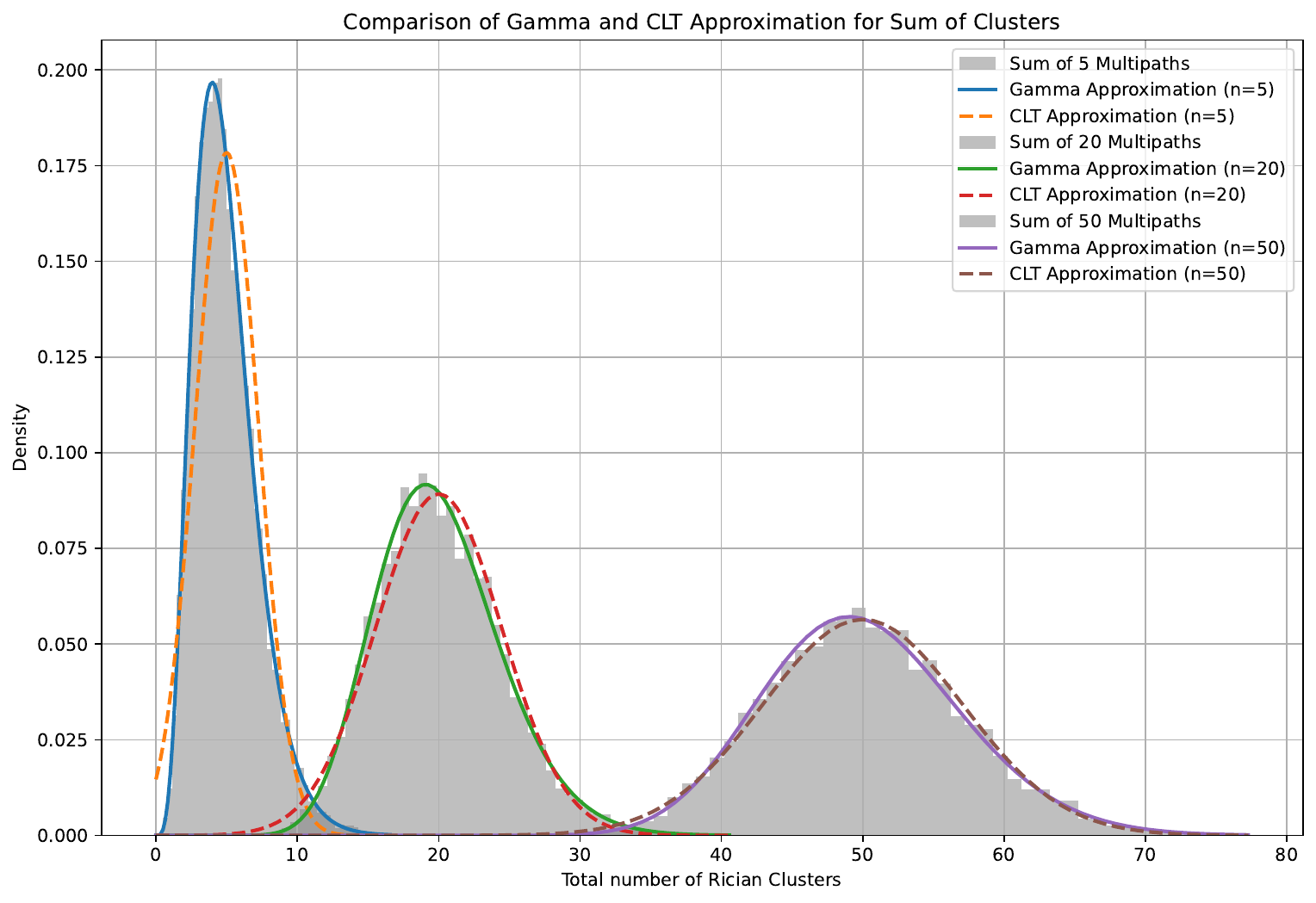}
	\vspace{-2mm}
	\caption{Comparison of the effectiveness of the Gaussian and Gamma distributions at different numbers of terms in the summation.}
	\vspace{-3mm}
	\label{fig:approximation_comparison}
\end{figure}

Using the Gamma distribution, the squared magnitude of the FAS-OTFS channel can be approximated as$ \lvert h_k \rvert^2 \sim \text{Gamma}(\alpha, \theta)$ where {we have:}%
\begin{align}\label{eqn:gamma_parameter_def}
	\alpha = \frac{\bigl(\mathbb{E}[\lvert h_k \rvert^2]\bigr)^2}{\mathrm{Var}(\lvert h_k \rvert^2)},
	\quad
	\theta = \frac{\mathrm{Var}(\lvert h_k \rvert^2)}{\mathbb{E}[\lvert h_k \rvert^2]}.
\end{align}

\noindent Therefore the first two moments of $|h_k|^2$ {have} to be derived in order to obtain the Gamma approximation. 
\begin{lemma}
	The first moment of $X = \lvert h_k \rvert^2$ is given by:%
	\begin{align}\label{eqn:first_moment}
	\mathbb{E}[\lvert h_k \rvert^2] = \sum_{p=1}^P \alpha_p^2 \bigl(\lvert \mu_{p,k} \rvert^2 + \sigma_{p,k}^2\bigr),
	\end{align}
	where $\alpha_p^2 = \tilde{\alpha}_{p} \delta(\tau - \tau_{p}) \delta(\nu - \nu_{p})$ is used to shorten the notation after combining (\ref{eqn:FAS-OTFS_channel}) and $c_{p,k}$.
\end{lemma}
\begin{proof}
 The first moment of $\lvert h_k \rvert^2$ is$ \mathbb{E}[\lvert h_k \rvert^2] = \mathbb{E}[h_k h_k^*]$. Upon expanding $h_k$:%
\begin{align}\label{eqn:h_k}
	h_k = \sum_{p=1}^P \alpha_p \bigl(\mu_{p,k} + \tilde{c}_{p,k}\bigr),
\end{align}

the expectation $ \mathbb{E}[\lvert h_k \rvert^2]$ becomes:%
\begin{align}\label{eqn:m1_derivation_1}
	\mathbb{E}[\lvert h_k \rvert^2] = \sum_{p=1}^P \alpha_p^2 \mathbb{E}\Bigl[\bigl(\mu_{p,k} + \tilde{c}_{p,k}\bigr) \bigl(\mu_{p,k} + \tilde{c}_{p,k}\bigr)^*\Bigr].
\end{align}

The expectation of the $p$-th cluster is:%
\begin{align}\label{eqn:m1_derivation_2}
	\mathbb{E}\Bigl[\bigl(\mu_{p,k} + \tilde{c}_{p,k}\bigr) \bigl(\mu_{p,k} + \tilde{c}_{p,k}\bigr)^*\Bigr] = \lvert \mu_{p,k} \rvert^2 + \sigma_{p,k}^2.
\end{align}
By substituting (\ref{eqn:m1_derivation_2}) into (\ref{eqn:m1_derivation_1}), (\ref{eqn:first_moment}) is obtained thus completing the proof.
\end{proof}	
\begin{theorem}
	The variance of $\lvert h_k \rvert^2$ without considering {any} correlation between {the} ports is (\ref{eqn:variance_no_correlation}) (top of next page).
	
	\begin{figure*}
		\vspace{-3mm}
		\begin{align}\label{eqn:variance_no_correlation}
			\tilde{\text{Var}}(|h_k|^2) = \Bigl\lvert \sum_{p=1}^P \alpha_p \mu_{p,k} \Bigr\rvert^4 +	4 \Bigl\lvert \sum_{p=1}^P \alpha_p \mu_{p,k} \Bigr\rvert^2
			\Bigl(\sum_{p=1}^P \alpha_p^2 \sigma_{p,k}^2\Bigr) + 2 \Bigl(\sum_{p=1}^P \alpha_p^2 \sigma_{p,k}^2\Bigr)^2 
			- \left( \sum_{p=1}^P \alpha_p^2 \bigl(\lvert \mu_{p,k} \rvert^2 + \sigma_{p,k}^2\bigr)  \right)^2 
		\end{align}
		\vspace{-2mm}
		\hrule
		\vspace{-3mm}
	\end{figure*}

\end{theorem}

\begin{proof}
	Refer to Appendix \ref{appendix:gamma_parameters}.
\end{proof}
 
\begin{lemma}
	The variance of $\lvert h_k \rvert^2$ {upon} considering {the} correlation between ports is (\ref{eqn:variance}) (top of next page), {where} $\eta \in [0.1,1]$ is the control parameter that prevents the over-estimation of {the} variance in scenarios with high inter-port correlation.
	
	\begin{figure*}
		\vspace{-3mm}
		\begin{align}\label{eqn:variance}
			\text{Var}(|h_k|^2) = \tilde{\text{Var}}(|h_k|^2) +  \min\left(\frac{2}{N^2} \sum_{i \neq j} J_0\left(\frac{2\pi |i - j| W}{N-1}\right) M_1^2, \eta \cdot \tilde{\text{Var}}(|h_k|^2)\right).
		\end{align}
		\vspace{-2mm}
		\hrule
		\vspace{-3mm}
	\end{figure*}
	
\end{lemma}

\begin{proof}
	Refer to Appendix \ref{appendix:correlation}.
\end{proof}
 
 \section{Performance Analysis: General Case}\label{sec:general}
 
{This section presents the derivations for outage probability and ergodic capacity.} Specifically, the optimal port is determined, {then} the SNR distribution at the optimal port is calculated, and finally the outage probability and ergodic capacity are derived.

Due to the complex nature of the distributions as well as the correlation  between ports, it is difficult to obtain a closed-form expression for both {the} SNR and {for} the {other} performance metrics. In order to obtain a closed-form expression, {a} bound and approximations are proposed for the SNR which are analysed using numerical simulations. 

\subsection{Optimal Port Selection and SNR}
The instantaneous SNR at port $k$ is defined as $\text{SNR}_k = |h_k|^2$ assuming normalized noise power. The SNR at the optimal port is:\vspace{-3mm}%
\begin{align}\label{eqn:snr}
	\text{SNR}_{\text{opt}} = \max_{1 \leq k \leq N} |h_k|^2.
\end{align}\vspace{-5mm}
\subsection{Distribution of $\text{SNR}_{\text{opt}}$}
The distribution of $\text{SNR}_{\text{opt}}$ is governed by the maximum of {the} correlated random variables $\{h_k\}$. Due to the {signal} correlation among ports, the distribution is non-trivial and typically requires numerical or simulation-based methods, {leading to} no closed-form expression. 

\subsubsection{Upper and Lower Bounds}

By observing two extreme scenarios, upper and lower bounds of the SNR can be defined.

\textbf{Scenario 1: Independent Ports}

A common upper bound {hinges on the} assumption {that} the ports $h_k$ are independent. The probability that all ports have $\lvert h_k \rvert^2 < \gamma$ is given by:%
\begin{align}
\Pr\!\bigl(\max_k \lvert h_k \rvert^2 < \gamma\bigr) 
= \prod_{k=1}^N F_k(\gamma).
\end{align}
If all ports {have} identical distributions, this simplifies to $F_{\mathrm{opt}}^{(\mathrm{uncorr})}(\gamma) = \bigl[F(\gamma)\bigr]^N.$ This overestimates the variability of the maxima because independence allows for {wider} spread, making this an upper bound on the true CDF:%
\begin{align}
F_{\mathrm{opt}}(\gamma) \leq F_{\mathrm{opt}}^{(\mathrm{uncorr})}(\gamma).
\end{align}
\textbf{Scenario 2: Fully Correlated Ports}
A lower bound assumes {that} the ports are perfectly correlated, meaning $\lvert h_1 \rvert^2 = \lvert h_2 \rvert^2 = \cdots = \lvert h_N \rvert^2$. In this extreme case $\max_k \lvert h_k \rvert^2 = \lvert h_1 \rvert^2$, and the CDF of $\mathrm{SNR}_{\mathrm{opt}}$ is simply the marginal CDF of any single port:%
\begin{align}
F_{\mathrm{opt}}^{(\mathrm{full\,corr})}(\gamma) = F(\gamma).
\end{align}
Since correlation reduces variability, the true maximum distribution satisfies:%
\begin{align}
F_{\mathrm{opt}}^{(\mathrm{full\,corr})}(\gamma) \geq F_{\mathrm{opt}}(\gamma).
\end{align}

%
%
%

\subsection{Outage Probability Expressions}

The exact outage probability is defined as:%
\begin{align}\label{eqn:outage_probability}
	P_{\text{out}}(\gamma_{\text{th}}) = \Pr(\text{SNR}_{\text{opt}} < \gamma_{\text{th}}).
\end{align}

The SNR at the optimal port in (\ref{eqn:snr}) can be adjusted to consider the impact of signal power, noise and antenna efficiency, {formulated as}: %
\begin{align}
	\text{SNR}_{\text{opt}} = \frac{P_{\text{transmitted}} \cdot G}{N_0} \max_k |h_k|^2,
\end{align}

\noindent where $P_{\text{transmitted}}$ is the transmit power, $G$ is the combined antenna gain and $N_0$ is the noise power spectral density. Given this, {upon} defining the SNR threshold $\gamma_{\text{th}}' = \frac{\gamma_{\text{th}} N_0}{P_{\text{transmitted}} G}$ in terms of power, the outage probability in (\ref{eqn:outage_probability}) can be written as:%
\begin{align}
	P_{\text{out}}(\gamma_{\text{th}}) = \Pr\left(\max_k |h_k|^2 < \gamma_{\text{th}}' \right).
\end{align}
\begin{theorem}
	The exact outage probability of a FAS-OTFS system is:%
	\begin{align}\label{eqn:outage_exact}
		P_{\text{out}}(\gamma_{\text{th}}) = \sum_{k=1}^{N} (-1)^{k+1} \sum_{S_k} \det(\mathbf{R}_{S_k}) \prod_{i \in S_k} F_{|h_k|^2} (\gamma_{\text{th}}'),
	\end{align}
	
	\noindent where $S_k$ represents all possible subsets of $k$ ports and $\mathbf{R}_{S_k}$ is the sub-matrix of $\mathbf{R}$ corresponding to those ports.
\end{theorem}

\begin{proof}
	Please refer to Appendix \ref{appendix:exact_outage}.
\end{proof}

\begin{theorem}
	The outage probability of an FAS-OTFS system is bounded by:%
	\begin{align}\label{eqn:outage_bounds}
		\frac{\gamma\bigl(\alpha, \frac{\gamma_{\text{th}}'}{\theta}\bigr)}{\Gamma(\alpha)}
		\leq
		F_{\mathrm{opt}}(\gamma)
		\leq
		\left[\frac{\gamma\bigl(\alpha, \frac{\gamma_{\text{th}}'}{\theta}\bigr)}{\Gamma(\alpha)}\right]^N.
	\end{align}

\end{theorem}
	
\begin{proof}
	Please refer to Appendix \ref{appendix:bounds}.
\end{proof}

\subsection{Asymptotic Bounds}

For specific ranges of $\gamma_{\mathrm{th}}$ (e.g., very small or very large), asymptotic approximations can simplify the outage probability bounds. By obtaining an expression for the single-port CDF $F_S(\gamma_{\mathrm{th}})$, this can then be substituted into the outage probability expressions. 

\subsubsection{Small $\gamma_{\mathrm{th}} / \theta$}

When \(\gamma_{\mathrm{th}} / \theta \ll 1\), the incomplete gamma function can be approximated as:%
\begin{align}
	\gamma(\alpha, \tfrac{\gamma_{\mathrm{th}}}{\theta}) \approx \frac{(\gamma_{\mathrm{th}} / \theta)^\alpha}{\alpha}.
\end{align}

Substituting this into the single-port CDF {yields}:%
\begin{align}\label{eqn:outage_asymptotic_small}
	F_S(\gamma_{\mathrm{th}}) \approx \frac{1}{\Gamma(\alpha)} \frac{(\gamma_{\mathrm{th}} / \theta)^\alpha}{\alpha}.
\end{align}

%
%

\subsubsection{ Large \(\gamma_{\mathrm{th}} / \theta\)}

When $\gamma_{\mathrm{th}} / \theta \gg 1$, the incomplete gamma function can be approximated as:%
\begin{align}
	\gamma(\alpha, \tfrac{\gamma_{\mathrm{th}}}{\theta}) \approx \Gamma(\alpha) - e^{-\gamma_{\mathrm{th}} / \theta} \frac{(\gamma_{\mathrm{th}} / \theta)^{\alpha - 1}}{\Gamma(\alpha)}.
\end{align}

Substituting this into the single-port CDF {of (\ref{eqn:outage_bounds}) yields}:%
\begin{align}
F_S(\gamma_{\mathrm{th}}) \approx 1 - e^{-\gamma_{\mathrm{th}} / \theta} \frac{(\gamma_{\mathrm{th}} / \theta)^{\alpha - 1}}{\Gamma(\alpha)}.
\end{align}

%
%

\subsection{Ergodic Capacity}


This subsection derives the ergodic capacity bounds for \(\mathrm{SNR}_{\mathrm{opt}} = \max_{1 \leq k \leq N} \lvert h_k \rvert^2\), using the previously established outage probability bounds. The ergodic capacity is defined as the expected value of the channel capacity:%
\begin{align}\label{eqn:capacity}
	C = \mathbb{E}[\log_2(1 + \gamma)] = \int_0^\infty \log_2(1 + \gamma) f_{\mathrm{opt}}(\gamma) \, d\gamma,
\end{align}

\noindent where \(f_{\mathrm{opt}}(\gamma)\) is the PDF of \(\mathrm{SNR}_{\mathrm{opt}}\), derived from the CDF \(F_{\mathrm{opt}}(\gamma)\).

\begin{lemma}
	The exact ergodic capacity for a FAS-OTFS system is: \vspace{-5mm}%
	\begin{align}\label{eqn:exact_capacity}
		C_{\text{erg}} = \frac{1}{\ln 2} \int_0^{\infty} \frac{1 - P_{\text{out}}(\gamma)}{1 + \gamma} d\gamma.
	\end{align}
\end{lemma}

\begin{proof}
	Since $P_{\text{out}}(\gamma)$ is directly derived in (\ref{eqn:outage_exact}), this can be substituted into (\ref{eqn:capacity}). Then, using integration by parts, this can be rewritten as (\ref{eqn:exact_capacity}).
\end{proof}

%
%

\subsection{ Capacity Bounds}

\subsubsection*{Lower Bound}
Substituting $f_S(\gamma)$ {from (\ref{eqn:outage_bounds})} directly into the ergodic capacity expression (\ref{eqn:capacity}) yields:%
\begin{align}\label{eqn:capacity_lower}
	C_{\mathrm{lower}} = \int_0^\infty \log_2(1 + \gamma) \frac{1}{\Gamma(\alpha) \theta^\alpha} \gamma^{\alpha-1} e^{-\frac{\gamma}{\theta}} \, d\gamma.
\end{align}

\subsubsection*{Upper Bound}
For the upper bound, substituting $N [F_S(\gamma)]^{N-1} f_S(\gamma)$ {from (\ref{eqn:outage_bounds})} for $f_{\mathrm{opt}}(\gamma)$ {in (\ref{eqn:capacity})} gives:\vspace{-2mm}%
\begin{align}\label{eqn:capacity_upper}
	C_{\mathrm{upper}} = \int_0^\infty \log_2(1 + \gamma)& N \left( \frac{\gamma\bigl(\alpha, \frac{\gamma}{\theta}\bigr)}{\Gamma(\alpha)} \right)^{N-1}\\ \nonumber 
	&\times \frac{1}{\Gamma(\alpha) \theta^\alpha} \gamma^{\alpha-1} e^{-\frac{\gamma}{\theta}} \, d\gamma.
\end{align}
\vspace{-5mm}%
\subsection{Special Case Approximations}
\subsubsection*{Small \(\gamma / \theta\)}
When $\gamma / \theta \ll 1$, {we have} $\log_2(1 + \gamma) \approx \frac{\gamma}{\ln(2)}$. Substituting this approximation {into (\ref{eqn:capacity_lower}) yields}:%
\paragraph{Lower Bound}%
\begin{align}
C_{\mathrm{lower}} \approx \frac{1}{\ln(2) \Gamma(\alpha) \theta^\alpha} \int_0^\infty \frac{\gamma^\alpha}{\theta} e^{-\frac{\gamma}{\theta}} \, d\gamma.
\end{align}
The integral evaluates to:%
\begin{align}
C_{\mathrm{lower}} \approx \frac{\alpha}{\ln(2)}.
\end{align}

\paragraph{Upper Bound}
For small \(\gamma\), \([F_S(\gamma)]^{N-1} \approx \left( \frac{\gamma}{\theta} \right)^{(N-1)\alpha}\), and the upper bound {in (\ref{eqn:capacity_upper})} becomes:%
\begin{align}
	C_{\mathrm{upper}} \approx \frac{N}{\ln(2)} \frac{1}{\Gamma(\alpha)} \int_0^\infty \gamma^{(N-1)\alpha + \alpha - 1} e^{-\frac{\gamma}{\theta}} \, d\gamma.
\end{align}

\subsubsection*{ Large \(\gamma / \theta\)}
When \(\gamma / \theta \gg 1\), \(\log_2(1 + \gamma) \approx \log_2(\gamma)\). Substituting this approximation {into (\ref{eqn:capacity}) yields}:%
\begin{align}
	C_{\mathrm{lower}} \approx \int_0^\infty \log_2(\gamma) f_S(\gamma) \, d\gamma.
\end{align}

The same applies to the upper bound with \(f_{\mathrm{opt}}(\gamma)\).

\section{Performance Analysis: Considering a Specific Scenario}\label{sec:specific}

Having provided bounds for outage probability and ergodic capacity, we provide a closed form expression for a special case. If a less general model was considered, a closed-form expression {may be possible} for both {the} outage and {the} capacity. By using a different model for Rician fading as well as applying new approximations, {a} closed-form expression {is obtained} for outage probability and ergodic capacity.

\subsection{Simplifying the Channel} 

The general channel expression for a FAS-OTFS channel was given in (\ref{eqn:FAS-OTFS_channel}). {Let us now} assume the specific scenario, where the LEO satellite is directly above the mobile device and also transmitting in a narrow beam. {This} results in a very strong LoS path from the satellite to the FAS, with the presence of any multipaths or clusters being negligible. By assuming {that} there is a single Rician cluster, this results in $P=1$ and (\ref{eqn:FAS-OTFS_channel}) can be re-written as:%
\begin{equation}
	h_k(\tau, \nu) =  \tilde{\alpha}_0 (\mu_{k} + \tilde{c}_{k}) \delta({\tau - \tau_0}) \delta({\nu - \nu_0}),
\end{equation}
 \noindent where $\tilde{\alpha}_0, \tau_0, \nu_0$ are the gain, delay and Doppler of the single LoS path. 
 \subsection{FAS Correlation Model}
 FAS Rayleigh channels have been defined in {the} previous literature as \cite{chai2022mlfluid}: %
 \begin{equation}
 	g_k = (\sqrt{1-\mu^2} x_k + \mu x_0) + j(\sqrt{1-\mu^2} y_k + \mu y_0),
 \end{equation}
 
\noindent where $\mu = \sqrt{2} \sqrt{ {}_1F_2 \left( \frac{1}{2}; 1, \frac{3}{2}; -\pi^2 W^2 \right) } - \frac{J_1(2\pi W)}{2\pi W} $ is the correlation between ports \cite{wong2022closedform} and $x_0, x_k, y_0, y_k$ are all Gaussian random variables with zero mean and a variance of $\frac{1}{2}$. {Furthermore}, ${}_1F_2 (\cdot)$ is the generalised hypergeometric function.
 
 In order to turn this FAS channel into a Rician fading channel, a deterministic complex component $\alpha_k$ is added: %
 \begin{align}
 	g_k &= \sqrt{\frac{K}{K+1}} \alpha_k \\ \nonumber 
 	&+ \sqrt{\frac{1}{K+1}} \left[ \sqrt{1-\mu^2} (x_k + jy_k ) + \mu(x_0 + jy_0)  \right]
 \end{align}
 
 By substituting $\mu_{p,k} + \tilde{c}_{p,k} = g_k$ {into (\ref{eqn:h_k})}, the overall FAS-OTFS channel can be obtained as $h_k (\tau, \nu) = \beta (\tau, \nu) g_k$, where $\beta (\tau, \nu) = \tilde{\alpha}_0  \delta({\tau - \tau_0}) \delta({\nu - \nu_0})$ is used for compactness. 
 
 \subsection{Conditioning on $x_0, y_0$}
 
 By fixing $x_0, y_0$, $h_k$ can be re-written as a sum of its constant and deterministic components: %
 \begin{align}
 	h_k =&  \left[ \alpha_k \beta \sqrt{\frac{K}{K+1}} +  \mu \beta \sqrt{\frac{1}{K+1}}(x_0 + jy_0) \right] \\ \nonumber 
 	&+ \beta \sqrt{\frac{1-\mu^2}{K+1}} (x_k + jy_k ).
 \end{align}
 
 With a deterministic and constant component, the pdf of $|h_k|$ would be a Rician distribution. For $r=|h_k|$ given $x_0,y_0$: %
 \begin{align}
 	f_{r \big|x_0,y_0}(r) = \frac{r}{\sigma^2}e^{- \frac{r^2 + |c_k|^2}{2\sigma^2}} \text{I}_0(\frac{r|c_k|}{\sigma^2}),
 \end{align}
 
 \noindent where {we have} $\sigma^2 = \beta^2 \frac{1-\mu^2}{K+1}$ and $c_k = \alpha_k \beta \sqrt{\frac{K}{K+1}} +  \mu \beta \sqrt{\frac{1}{K+1}}(x_0 + jy_0)$. conditioned on $x_0, y_0$, $|h_1|,...,|h_N|$ are all independent and the joint pdf can be written as: %
 \begin{align}\label{eqn:conditioned_pdf}
 	f_{|h_1|,...,|h_N| \big|x_0,y_0}(r_1, ..., r_N) = \prod^N_{k=1} \left[ f_{|h_k| \big|x_0,y_0}(r_k)  \right].
 \end{align}

 \begin{theorem}
 	The unconditioned joint pdf of $|h_k|$ is {given by}  (\ref{eqn:pdf_unconditioned}) (top of next page).
 		\begin{figure*}
 			\vspace{-3mm}
 		\begin{align}\label{eqn:pdf_unconditioned}
 			f_{|h_1|,...,|h_N| }(r_1, ..., r_N) = \frac{1}{2\pi} \int^{\infty}_{\infty} \int^{\infty}_{\infty}	 \prod^N_{k=1} \left[ \frac{r}{\sigma^2}e^{- \frac{r^2 + |c_k|^2}{2\sigma^2}} \text{I}_0(\frac{r|c_k|}{\sigma^2})  \right]  \frac{1}{2\pi} e^{- \frac{x^2 + y^2}{2}} dx dy.
 		\end{align}
 		\vspace{-2mm}
 		\hrule
 		\vspace{-3mm}
 	\end{figure*}
 \end{theorem}
 
 \begin{proof}
 	The unconditioned pdf of $|h_k|$ can be obtained by integrating the product of the conditional pdf and the joint pdf of $x_0, y_0$ with respect to $(x_0,y_0)$, {yielding}: %
 	\begin{align}\label{eqn:pdf_unconditioned_definition}
 		&f_{|h_1|,...,|h_N| }(r_1, ..., r_N) = \\ \nonumber 
 		&\int^{\infty}_{\infty} \int^{\infty}_{\infty}	f_{|h_1|,...,|h_N| \big|x_0,y_0}(r_1, ..., r_N) f_{x_y,y_0}(x,y) dx dy.
 	\end{align}
 	
 	Since $x_0,y_0$ are normally distributed with zero mean and a variance of $\frac{1}{2}$, the joint pdf is: %
 	\begin{equation}\label{eqn:pdf_x0_y0}
 		f_{x_0,y_0}(x,y) = \frac{1}{2\pi} e^{- \frac{x^2 + y^2}{2}}.
 	\end{equation}
 \end{proof}
 By substituting (\ref{eqn:conditioned_pdf}) and (\ref{eqn:pdf_x0_y0}) into (\ref{eqn:pdf_unconditioned_definition}) the unconditioned joint pdf of $\lvert h_k \rvert$ is obtained.
 \subsection{Outage Probability Expressions}
 The outage probability can be expressed as {seen in} (\ref{eqn:outage_probability_def}) (top of next page).
 \begin{figure*}
 	\vspace{-3mm}
 	\begin{align}\label{eqn:outage_probability_def}
 		p(\gamma_{th}) = P\!\left(\max_k |h_k| \leq \gamma_{\mathrm{th}} )\right) = \int^{\gamma_{th}}_0 \dots \int^{\gamma_{th}}_0 f_{|h_1|,...,|h_N| }(r_1, ..., r_N) dr_1 \dots dr_N.
 	\end{align}
 	\vspace{-2mm}
 	\hrule
 	\vspace{-3mm}
 \end{figure*}
 
 \begin{theorem}
 	The exact outage probability of a FAS-OTFS system considering a single LoS path is (\ref{eqn:special_outage_probability_exact}) (top of next page).
 	\begin{figure*}
 		\vspace{-3mm}
 		\begin{align}\label{eqn:special_outage_probability_exact}
 			p(\gamma_{th}) = \int^{\infty}_{\infty} \int^{\infty}_{\infty}	 \prod^N_{k=1} \left[ 1 - \text{Q}_1 \left(  \frac{|c_k (x,y)|}{\sigma}, \frac{\gamma_{th}}{\sigma}  \right) \right]  \frac{1}{2\pi} e^{- \frac{x^2 + y^2}{2}} dx dy,
 		\end{align}
 		\vspace{-2mm}
 		\hrule
 		\vspace{-3mm}
 	\end{figure*}
 	where $\text{Q}_1(\cdot)$ is the Marcum-Q function of the first order.
 \end{theorem}
 \begin{proof}
 	Please refer to Appendix \ref{appendix:exact_outage_special}.
 \end{proof}
 
 In order to obtain an approximation for the closed form solution, the Gauss-Hermite quadrature \cite{abramowitz1972handbook} can be used, which takes the form: %
 \begin{align}\label{eqn:gauss_hermite}
 	 \int_{-\infty}^\infty e^{-t^2} f(t) \, \mathrm{d}t \approx \sum_{m=1}^M w_m f(x_m),
 \end{align}
 where $ M$ is the number of quadrature points (nodes and weights), $ \{x_m\}_{m=1}^M $ are the nodes (precomputed roots of the Hermite polynomial $ H_M(x) $) and $ \{w_m\}_{m=1}^M $ are the weights (derived from the Hermite polynomials).
 
 \begin{theorem}
 	The approximation of the closed-form solution of the outage probability via the Gauss-Hermite quadrature is (\ref{eqn:quadrature_outage}) (top of next page), where $ M $ is the number of quadrature points in each dimension, $\{x_m, x_n\} $ are the nodes and $ \{w_m, w_n\} $ are the corresponding weights.
 		\begin{figure*}
 			\vspace{-3mm}
 		\begin{align}\label{eqn:quadrature_outage}
 			P(\text{outage}) \approx \frac{1}{\pi} \sum_{m=1}^M \sum_{n=1}^M w_m w_n \prod_{k=1}^N \left[1 - Q_1\left(\frac{|c_k(\sqrt{2} x_m, \sqrt{2} x_n)|}{\sigma}, \frac{\gamma_{\mathrm{th}}}{\sigma}\right)\right],
 		\end{align}
 		\vspace{-2mm}
 		\hrule
 		\vspace{-3mm}
 	\end{figure*}
 \end{theorem}
 \begin{proof}
 	Please refer to Appendix \ref{appendix:quadrature_outage}.
 \end{proof}
\subsection{Ergodic Capacity Expressions}
\begin{lemma}
	Given the definition {of} ergodic capacity in (\ref{eqn:capacity}), by substituting the outage probability expression (\ref{eqn:special_outage_probability_exact}) {into (\ref{eqn:capacity})} the exact ergodic capacity of a FAS-OTFS system considering a single LoS path can be obtained, {see} (\ref{eqn:special_capacity_exact}) (top of next page).
	
	\begin{figure*}
		\vspace{-3mm}
		\begin{align}\label{eqn:special_capacity_exact}
			C = \frac{1}{\ln(2)} \int_{-\infty}^\infty \int_{-\infty}^\infty p_{(x, y)}(x, y) \left[ \int_0^\infty \frac{\prod_{k=1}^N Q_1\left( \frac{|c_k(x, y)|}{\sigma}, \frac{\sqrt{\gamma}}{\sigma} \right)}{1 + \gamma} \, d\gamma \right] dx \, dy.
		\end{align}
		\vspace{-2mm}
		\hrule
		\vspace{-3mm}
	\end{figure*}
\end{lemma}

However, this expression once again requires numerical evaluation. Therefore to find a closed-form approximation, {once again} the Gauss-Hermite quadrature is applied. The Gauss-Laguerre quadrature \cite{abramowitz1972handbook} is also applied due to the presence of integrals that are not bounded by $[-\infty, \infty]$.
 
\begin{theorem}
	The approximation of the closed form solution of the ergodic capacity via the Gauss-Hermite and Gauss-Laguerre quadratures is (\ref{eqn:quadrature_capacity}) (top of page 10).
	\begin{figure*}
		\vspace{-3mm}
		\begin{align}\label{eqn:quadrature_capacity}
			C \approx \frac{2}{\ln(2)} \sum_{m=1}^M \sum_{n=1}^M w_m w_n \left[ \sum_{l=1}^L w_l \prod_{k=1}^N Q_1\left( \frac{|c_k(\sqrt{2} x_m, \sqrt{2} x_n)|}{\sigma}, \frac{\sqrt{t_l}}{\sigma} \right) e^{t_l - \ln(1 + t_l)} \right].
		\end{align}
		\vspace{-2mm}
		\hrule
		\vspace{-3mm}
	\end{figure*}
\end{theorem}

\begin{proof}
	Please refer to Appendix \ref{appendix:quadrature_capacity}.
\end{proof}

\section{Numerical Results and Discussion}\label{sec:results}
\subsection{Outage Probability Model for MRC-OTFS Benchmark}
To benchmark the performance of our proposed FAS-OTFS system, we adopt a Maximal Ratio Combining (MRC) receiver for OTFS and derive its outage probability. Following the outage probability framework in \cite{zhang2024outage}, we extend the model to an MRC diversity system. The post-combining signal-to-noise ratio (SNR) is given by:%
\vspace{-2mm}
\begin{equation}
	\gamma_{\text{MRC}} = \frac{E_s}{N_0} \sum_{n=1}^{N_r} \|\mathbf{H}_n\|^2
\end{equation}
\noindent where $E_s$ is the symbol energy, $N_0$ is the noise power spectral density, and $\|\mathbf{H}_n\|^2$ represents the channel gain for the $n$-th receive antenna. The outage probability is then expressed as: %
\begin{align}\label{eqn:MRC}
	P_{\text{out}} &= \Pr \left( \gamma_{\text{MRC}} < \gamma_{\text{th}} \right) = F_{\chi^2_{2N_r}} \left( \frac{\gamma_{\text{th}} N_0}{E_s} \right)
\end{align}
\noindent where $\gamma_{\text{th}} = 2^R - 1$ is the SNR threshold for a target rate $R$, and $F_{\chi^2_{2N_r}}(\cdot)$ denotes the cumulative distribution function (CDF) of a chi-squared variable with $2N_r$ degrees of freedom.
%
\vspace{-5mm}
\subsection{Refinement of the Outage Probability Lower Bound}
The lower bound on the outage probability provides a worst-case estimate {of the} system's reliability for the general model. This bound was derived under the {idealized} assumption of fully correlated ports, all ports sharing identical fading conditions. {However}, this bound is often overly conservative, leading to a significant gap compared to numerical results, as it does not account for {realistic} correlation effects. {Indeed}, during the simulations it was noted that this lower bound was so low that it provided no insight. {Hence it was excluded from} the figures. Therefore, refinement was required.

To refine this bound, an effective number $N_{\mathrm{eff}}$ of independent ports are introduced, which interpolates between the fully correlated and fully independent cases:%
\begin{align}
	N_{\mathrm{eff}} = 1 + (N - 1) (1 - \rho_{\mathrm{avg}}),
\end{align}

\noindent where $\rho_{\mathrm{avg}}$ is the average correlation coefficient between ports. This parameter governs the correlation level:
\begin{itemize}
	\item If $\rho_{\mathrm{avg}} = 1$, then $N_{\mathrm{eff}} = 1$, {representing} the original fully correlated bound.
	\item If $\rho_{\mathrm{avg}} = 0$, then $N_{\mathrm{eff}} = N$, approaching the independent case.
\end{itemize}

The refined lower bound is then expressed as:%
\begin{align}
	P_{\mathrm{out}}^{(\mathrm{lower, improved})} = F_S(\gamma_{\mathrm{th}})^{N_{\mathrm{eff}}} = \left( \frac{\gamma\big(\alpha, \frac{\gamma_{\mathrm{th}}}{\theta} \big)}{\Gamma(\alpha)} \right)^{N_{\mathrm{eff}}}.
\end{align}

This formulation maintains theoretical validity, while significantly improving approximation accuracy. This revised bound is used in all lower bound plots for the general model.
\vspace{-2mm}
\subsection{General Model}
%
%
%
\begin{figure}[!t]
	\centering
	\includegraphics[width=3in]{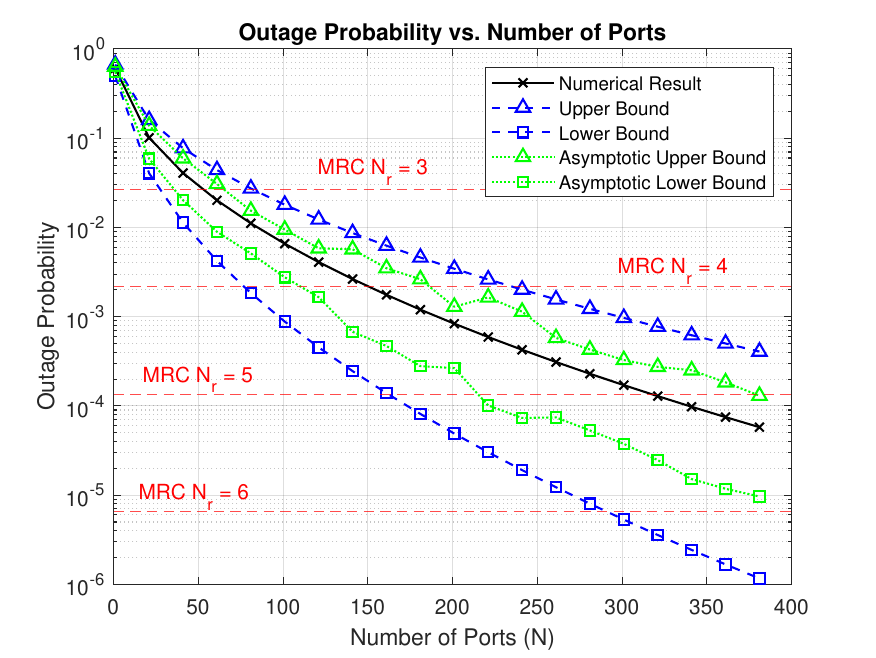}
	\vspace{-2mm}
	\caption{Outage probability versus the number of FAS ports ($N$) for the FAS-OTFS system, including numerical results and asymptotic bounds.}
	\vspace{-3mm}
	\label{fig:OPvN_general}
\end{figure}
%

Figure \ref{fig:OPvN_general} illustrates the impact of increasing the number $N$ of FAS ports on the outage probability in an OTFS-based system. As expected, increasing $N$ significantly reduces the outage probability, confirming the diversity gain provided by multiple ports. The numerical results closely follow the theoretical bounds, validating the {assumptions in the} analytical framework. The asymptotic upper and lower bounds provide additional insight into the performance limits, particularly at high $N$, where the outage probability approaches a floor dictated by system constraints. {Comparing with the MRC threshold in (\ref{eqn:MRC}), it can be observed that the general model outperforms a $N_r = 4$ antenna MRC system given a sufficiently large number of ports. Given the assumption in the FAS-OTFS model of an antenna size of $1\lambda$ during simulations, FAS-OTFS not only has a better outage performance but also is more energy and space efficient than similar MRC systems.}
\begin{figure}[!t]
	\centering
	\includegraphics[width=3in]{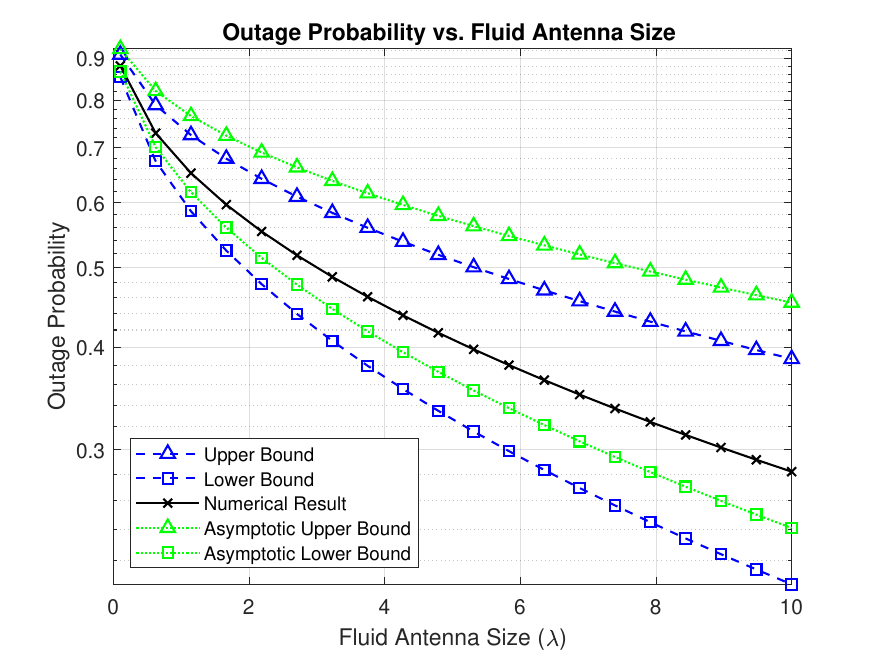}
	\vspace{-2mm}
	\caption{Outage probability versus the fluid antenna size ($W$) for the FAS-OTFS system, including numerical results and asymptotic bounds.}
	\vspace{-3mm}
	\label{fig:OPvW_general}
\end{figure}

In contrast to Figure \ref{fig:OPvN_general}, Figure \ref{fig:OPvW_general} examines how the FAS size $W$ influences the outage probability. A larger $W$ improves performance by offering {higher} spatial diversity and allowing the system to select more favourable channel conditions. However, diminishing returns can be observed at large $W$, where additional increases yield marginal improvements. The asymptotic bounds capture this behaviour well, demonstrating their utility in approximating the system performance across different spatial configurations.
\begin{figure}[!t]
	\centering
	\includegraphics[width=3in]{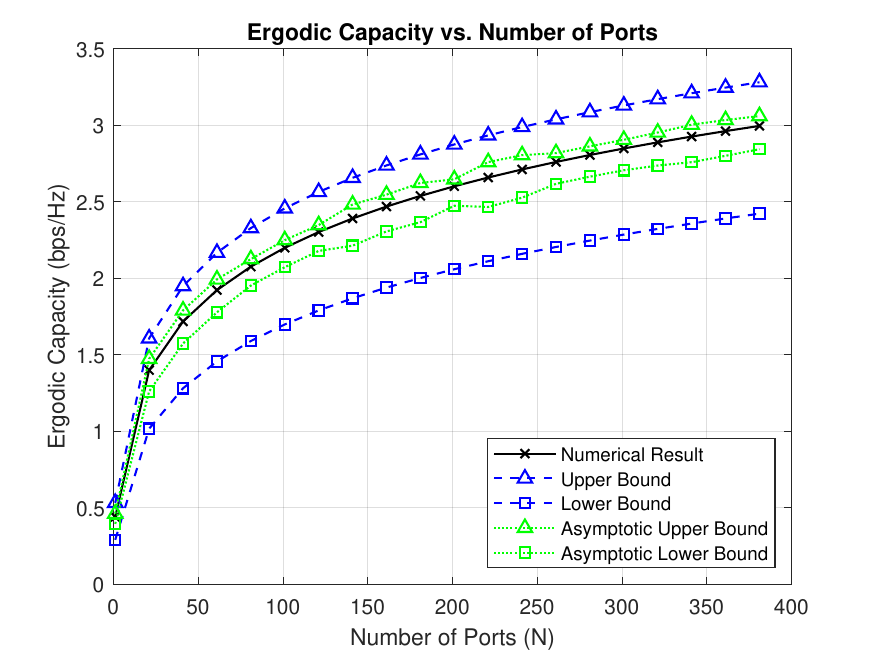}
	\vspace{-2mm}
	\caption{Ergodic capacity versus the number of FAS ports ($N$) for the FAS-OTFS system, including numerical results and asymptotic bounds.}
	\vspace{-3mm}
	\label{fig:ECvN_general}
\end{figure}

Figure \ref{fig:ECvN_general} highlights the ergodic capacity improvements as the number of ports $N$ increases. The capacity increases rapidly for small $N$ due to the significant diversity gain, but {its gradient gradually reduces} at higher $N$, indicating {the impact} imposed by channel conditions and system resources. The lower bound provides a conservative estimate, while the upper bound remains sufficiently tight, reinforcing the robustness of proposed analytical framework. This emphasizes the performance vs. complexity trade-off in practical implementations.
\begin{figure}[!t]
	\centering
	\includegraphics[width=3in]{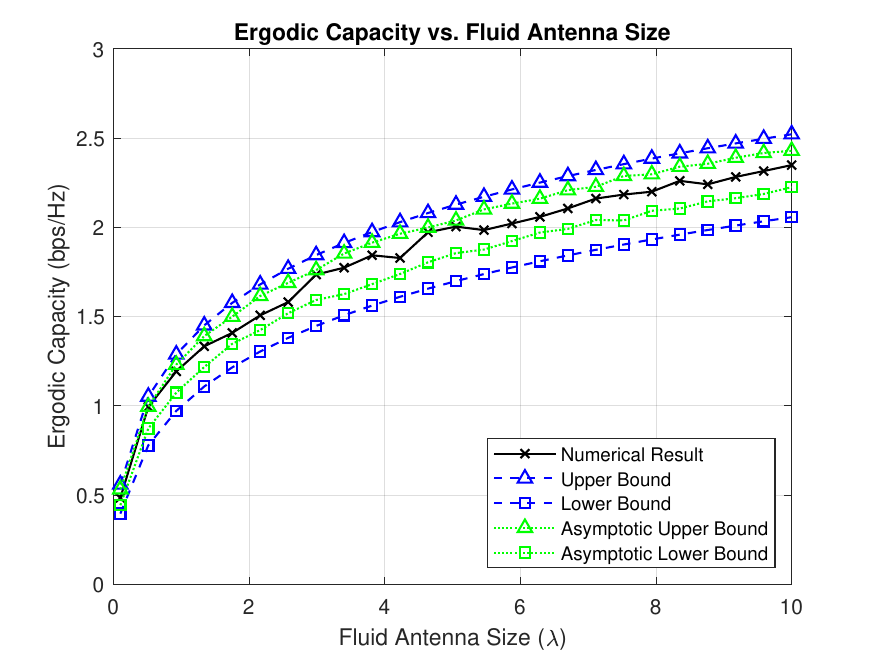}
	\vspace{-2mm}
	\caption{Ergodic capacity versus the fluid antenna size ($W$) for the FAS-OTFS system, including numerical results and asymptotic bounds.}
	\vspace{-3mm}
	\label{fig:ECvW_general}
\end{figure}

{In contrast} to Figure \ref{fig:ECvN_general}, Figure \ref{fig:ECvW_general} explores how the ergodic capacity scales with the fluid antenna size $W$. Increasing $W$ leads to improved spectral efficiency by leveraging better channel conditions, yet a saturation effect is visible beyond a certain $W$, emphasizing practical deployment limitations. The asymptotic bounds provide a useful approximation, particularly at larger $W$. This figure reinforces the importance of {carefully} balancing antenna size and system complexity to achieve the capacity gains desired.
\vspace{-2mm}
\subsection{Single Path Scenario}
\begin{figure}[!t]
\centering
\includegraphics[width=3in]{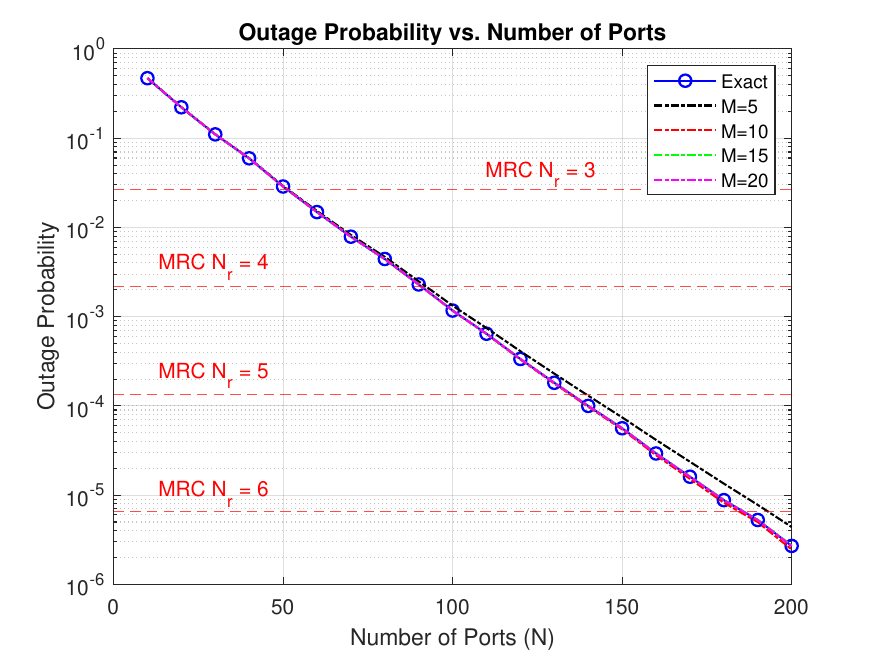}
\vspace{-2mm}
\caption{Outage probability versus the total number of ports $N$ for the special single-path FAS-OTFS scenario. Theoretical results (Exact) are compared with different values of $M$, the precision of the Gaussian quadrature.}
\vspace{-3mm}
\label{fig:OPvsN_special}
\end{figure}

Figure \ref{fig:OPvsN_special} illustrates the outage probability of the FAS-OTFS system under a single-path fading scenario, plotted against the number of FAS ports $N$. The exact analytical results are depicted by blue circles, while numerical evaluations using Gaussian quadratures with different precision levels $(M=5,10,15,20)$ are represented by dashed lines. The results indicate a steep {reduction} in outage probability as $N$ increases, showcasing the reliability improvements enabled by increasing the number of FAS ports. Moreover, the close alignment between the exact solution and numerical evaluations validates the outage probability expressions derived and confirms the accuracy of the Gaussian quadrature approach. {Furthermore, the single-path assumption has even better performance compared to the benchmark MRC scheme compared to the general model in Figure \ref{fig:OPvN_general}, with the single-path model outperforming even $N_r = 6$ MRC.} These findings reinforce the effectiveness of FAS-OTFS for energy-efficient and reliable IoT communications in high-mobility satellite environments.
\begin{figure}[!t]
	\centering
	\includegraphics[width=3in]{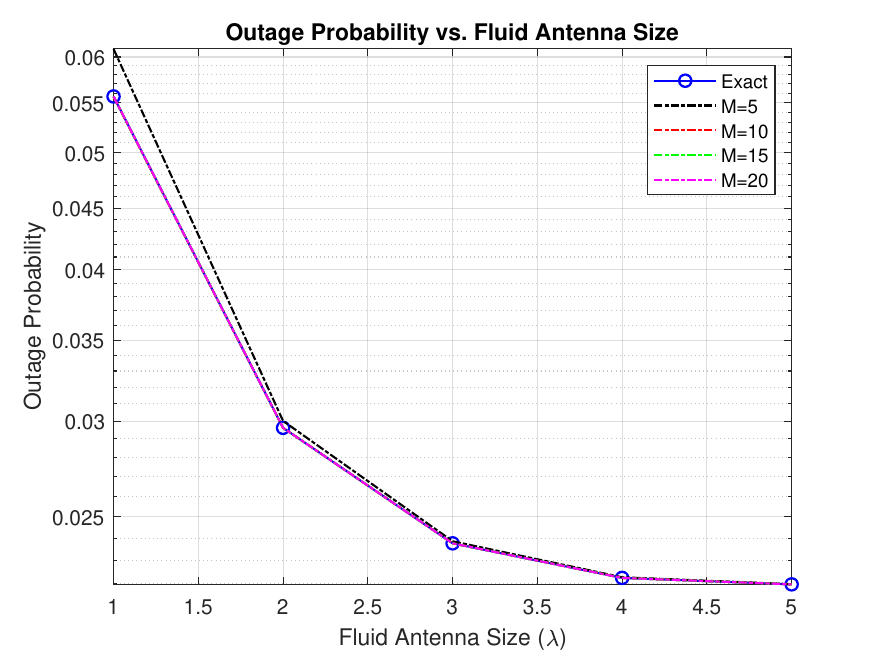}
	\vspace{-2mm}
	\caption{Outage probability versus the FAS size $W$ in terms of wavelength $\lambda$ for the special single-path FAS-OTFS scenario. Theoretical results (Exact) are compared with numerical evaluations using Gaussian quadrature with different precision levels ($M = 5, 10, 15, 20$).}
	\vspace{-3mm}
	\label{fig:OPvsW_special}
\end{figure}

Figure \ref{fig:OPvsW_special} illustrates the relationship between the FAS size $W$, measured in terms of wavelength $\lambda$ and the outage probability of the FAS-OTFS system in a special single-path fading scenario. The exact analytical results are represented by blue circles, while {our} numerical evaluations using Gaussian quadratures with different precision levels ($M = 5, 10, 15, 20$) are shown with dashed lines. The results demonstrate that as the FAS size increases, the outage probability significantly decreases, indicating improved system reliability for larger spatial apertures. Furthermore, the close match between the exact results and numerical evaluations validates the accuracy of the outage probability expressions derived. These findings emphasize the importance of FAS size optimization in enhancing the performance of FAS-OTFS for energy-efficient, high-mobility satellite IoT communications.
\begin{figure}[!t]
	\centering
	\includegraphics[width=3in]{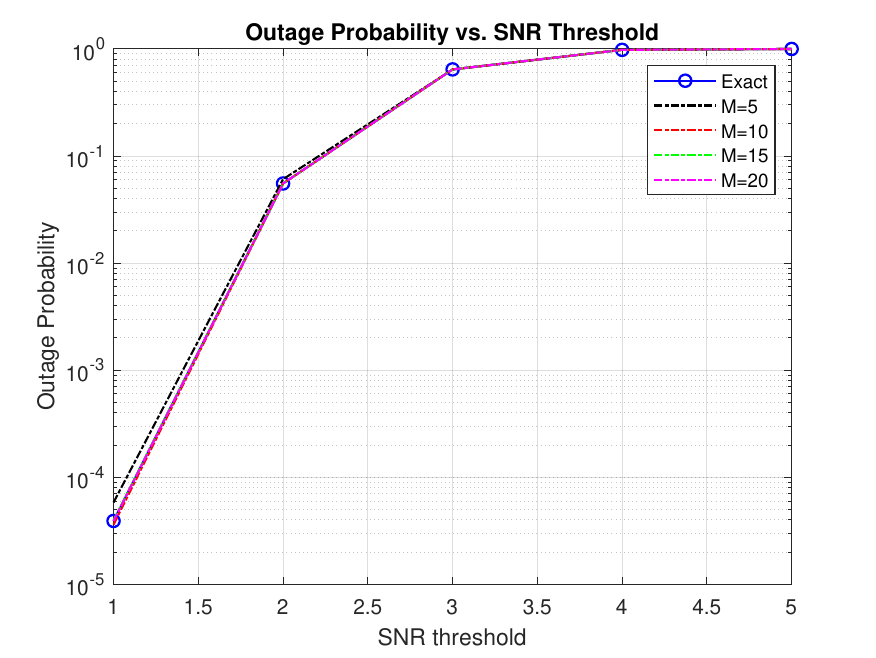}
	\vspace{-2mm}
	\caption{Outage probability versus the signal-to-noise ratio (SNR) threshold for the special single-path FAS-OTFS scenario. Theoretical results (Exact) are compared with numerical evaluations using Gaussian quadrature with different precision levels ($M = 5, 10, 15, 20$).}
	\vspace{-2mm}
	\label{fig:OPvsSNRth_special}
\end{figure}

This figure depicts the effect of the SNR threshold on the outage probability of the FAS-OTFS system in a {specific} single-path fading scenario. The exact analytical results are shown by blue circles, while {our} numerical evaluations using Gaussian quadrature with varying precision levels ($M = 5, 10, 15, 20$) are plotted with dashed lines. The results reveal that the outage probability increases steeply as the SNR threshold increases, indicating the growing difficulty in maintaining reliable communication {in the face of} higher SNR requirements. Additionally, the strong agreement between the exact and numerical results confirms the accuracy of derived probability expressions. These highlight the trade-off between reliability and the performance thresholds of FAS-OTFS systems for energy-efficient satellite IoT communications.
\begin{figure}[!t]
	\centering
	\includegraphics[width=3in]{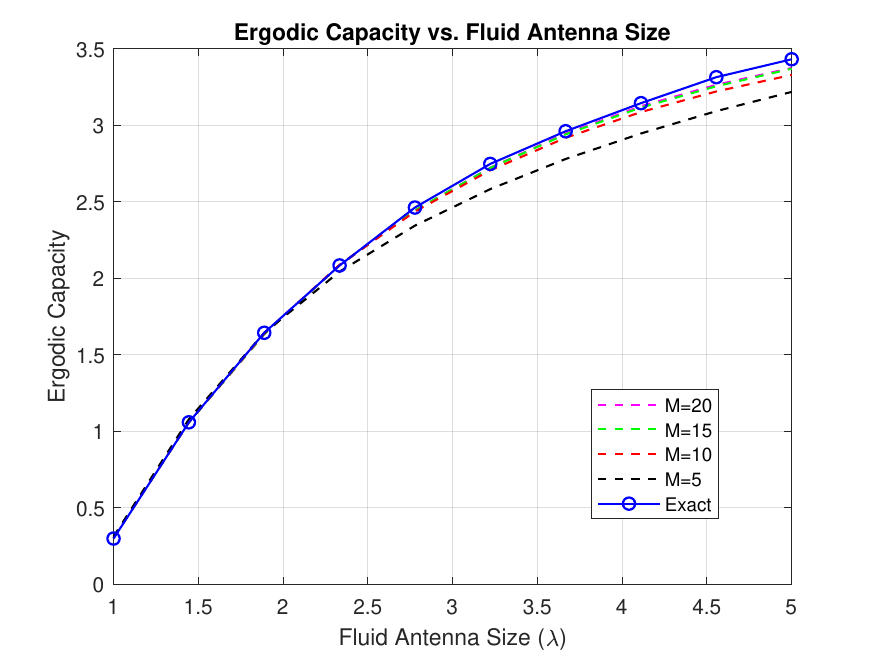}
	\vspace{-2mm}
	\caption{Ergodic capacity versus fluid antenna size $W$ (in $\lambda$), showing exact values (blue) and approximations using Gauss-Laguerre and Gauss-Hermite quadrature for different $M$, the precision of the quadratures.}
	\vspace{-2mm}
	\label{fig:ECvW_special}
\end{figure}

Figure \ref{fig:ECvW_special} illustrates the relationship between ergodic capacity and the fluid antenna size $W$, demonstrating that as $W$ increases, the capacity improves significantly. The exact values are closely followed by approximations using Gauss-Laguerre and Gauss-Hermite quadratures with varying $M$, where a higher $M$ provides better accuracy. This trend indicates that larger fluid antenna sizes enhance the signal diversity and improve capacity, with diminishing gains at higher $W$. It can be observed that for a larger FAS, the approximations become less accurate. However, due to space constraints on IoTDs, the size of the FAS would normally be $\leq 5\lambda$ and this is not a significant issue.
\begin{figure}[!t]
	\centering
	\includegraphics[width=3in]{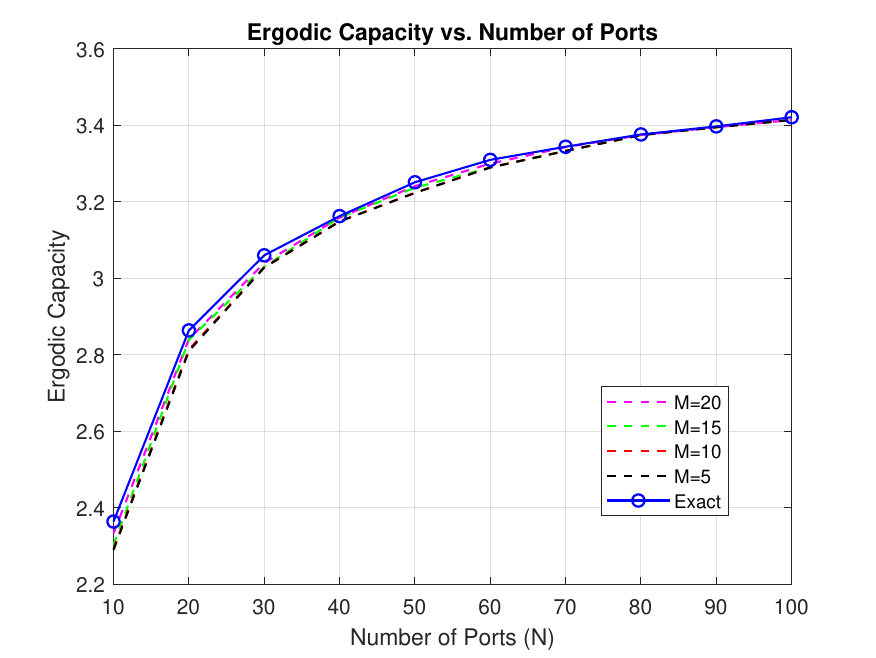}
	\vspace{-2mm}
	\caption{Ergodic capacity versus number of ports $N$, comparing exact values (blue) and quadrature approximations for different $M$.}
	\vspace{-2mm}
	\label{fig:ECvN_special}
\end{figure}

Figure \ref{fig:ECvN_special} presents the ergodic capacity as a function of the number of ports $N$, showing that the capacity increases rapidly for small $N$, before {exhibiting a gradually} saturating {tendency} at higher values. The approximations with different $M$ align well with the exact values, confirming their accuracy. {Again}, increasing $N$ improves the capacity, but the benefits {gradually} become marginal beyond a certain point, highlighting a {beneficial} range for practical system design.

%
\section{Conclusion}\label{sec:conclude}
{We conceived a} novel FAS-OTFS framework for enhancing {both the} energy and computational efficiency of satellite-based IoT communications in high-mobility environments. By integrating FAS with OTFS modulation, we addressed the computational limitations of IoT devices while maintaining robust performance under severe delay-Doppler impairments. We derived analytical expressions for {both the} outage probability and ergodic capacity under a general FAS-OTFS channel, with closed-form solutions obtained for a special single-path scenario. {Our} numerical results validate the effectiveness of FAS-OTFS, demonstrating significant improvements both in reliability and capacity over conventional OTFS. These findings establish FAS-OTFS as a promising solution for energy-efficient, high-mobility IoT communications over satellite links. Future research may explore {joint channel and data estimation} {relying on bespoke channel coding and turbo detection for approaching our performance limits derived}.


\appendices
\vspace{-3mm}
\section{Derivation of the Second Moment without correlation}
\label{appendix:gamma_parameters}

The second moment of the power $\lvert h_k \rvert^2$ for a multi-cluster Rician channel is derived, approximating the summation as a single Rician variable upon assuming uncorrelated diffuse components across clusters.
\vspace{-5mm}
\subsection{Combining Means and Variances (Approximation)}

\subsubsection*{Total Deterministic Component (Mean)}
The total specular component (mean) across all clusters is:%
\begin{align}
	M \triangleq \sum_{p=1}^P \alpha_p \mu_{p,k}.
\end{align}
\subsubsection*{ Total Diffuse Variance}
Assume that \(\{\tilde{c}_{p,k}\}_{p=1}^P\) are uncorrelated across different clusters. Then the sum of diffuse components becomes:%
\vspace{-2mm}
\begin{align}
Y \triangleq \sum_{p=1}^P \alpha_p \tilde{c}_{p,k},
\end{align}
which is a zero-mean complex Gaussian variable:%
\begin{align}
Y \sim \mathcal{CN}\bigl(0, \Sigma^2\bigr),
\quad
\Sigma^2 = \sum_{p=1}^P \alpha_p^2 \sigma_{p,k}^2.
\end{align}
\subsubsection*{Rician Approximation}
The channel can now be approximated as $h_k = M + Y$, where $M$ is deterministic, and $Y \sim \mathcal{CN}(0, \Sigma^2)$. Thus, $h_k$ is a Rician random variable, and $\lvert h_k\rvert^2$ is the power of a Rician random variable.
\vspace{-5mm}
\subsection{ Known Fourth Moment of a Rician Random Variable}
For a Rician random variable $X = m + Z$, where $Z \sim \mathcal{CN}(0, \sigma^2)$, the fourth moment of $\lvert X \rvert^2$ is:%
\begin{align}
	\mathbb{E}\bigl[\lvert X \rvert^4\bigr]	= \lvert m \rvert^4 + 4 \lvert m \rvert^2 \sigma^2 +	2 \sigma^4.
\end{align}

\vspace{-8mm}
\subsection{Applying to $h_k = M + Y$}
For the multi-cluster Rician channel:%
\begin{align} \label{eqn:fourth_moment_values}
	M =	\sum_{p=1}^P \alpha_p \mu_{p,k},
	\quad
	\Sigma^2 =\sum_{p=1}^P \alpha_p^2 \sigma_{p,k}^2.
\end{align}
Substituting (\ref{eqn:fourth_moment_values}) into the Rician fourth-moment formula yields $\mathbb{E}[\lvert h_k \rvert^4]	= \lvert M \rvert^4 + 4 \lvert M \rvert^2 \Sigma^2 + 2 \Sigma^4$. Expanding $M$ and $\Sigma^2$ {leads to}:
\vspace{-5mm}
\begin{align}\label{eqn:second_moment}
	\mathbb{E}[\lvert h_k \rvert^4]	= \Bigl\lvert \sum_{p=1}^P \alpha_p \mu_{p,k} \Bigr\rvert^4 &+ 4 \Bigl\lvert \sum_{p=1}^P \alpha_p \mu_{p,k} \Bigr\rvert^2 \Bigl( \sum_{p=1}^P \alpha_p^2 \sigma_{p,k}^2 \Bigr) \nonumber\\ 
	 &+ 2 \Bigl( \sum_{p=1}^P \alpha_p^2 \sigma_{p,k}^2 \Bigr)^2.
\end{align}
\vspace{-10mm}
\subsection{Obtaining the uncorrelated variance}
By substituting (\ref{eqn:first_moment}) and (\ref{eqn:second_moment}) into the variance $\tilde{\text{Var}}(\lvert h_k \rvert^2) = \mathbb{E}[\lvert h_k \rvert^4] - \mathbb{E}[\lvert h_k \rvert^2]^2$, (\ref{eqn:variance_no_correlation}) is obtained. Substituting the variance into (\ref{eqn:gamma_parameter_def}), simulations of the Gamma approximation and the original channel can be compared to see the accuracy of the approximation. Figure \ref{fig:no_correlation} shows that the moments derived for the non-correlated Gamma approximation are valid.

\begin{figure}[!t]
	\centering
	\includegraphics[width=0.49\columnwidth]{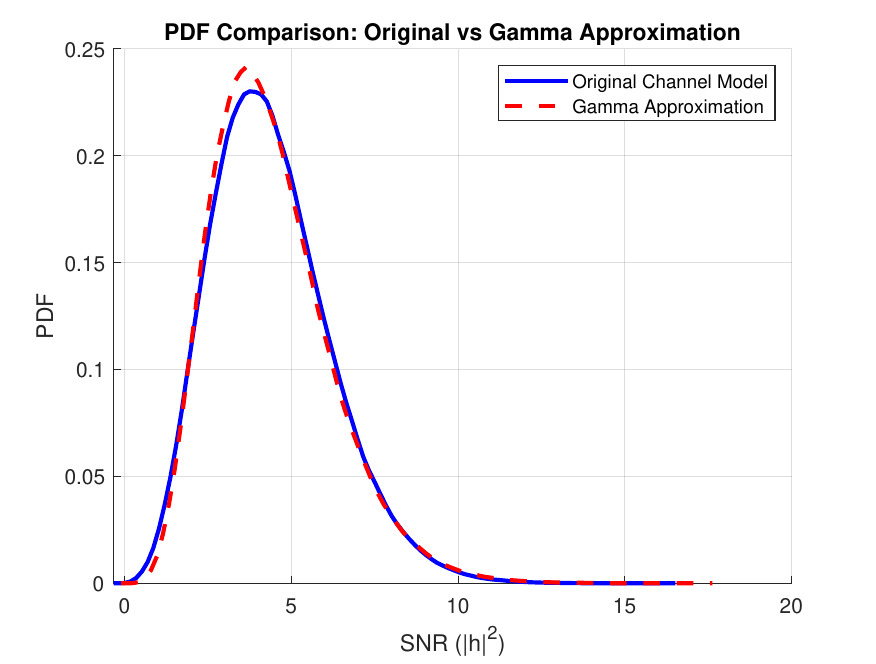}
	\hfill
	\includegraphics[width=0.49\columnwidth]{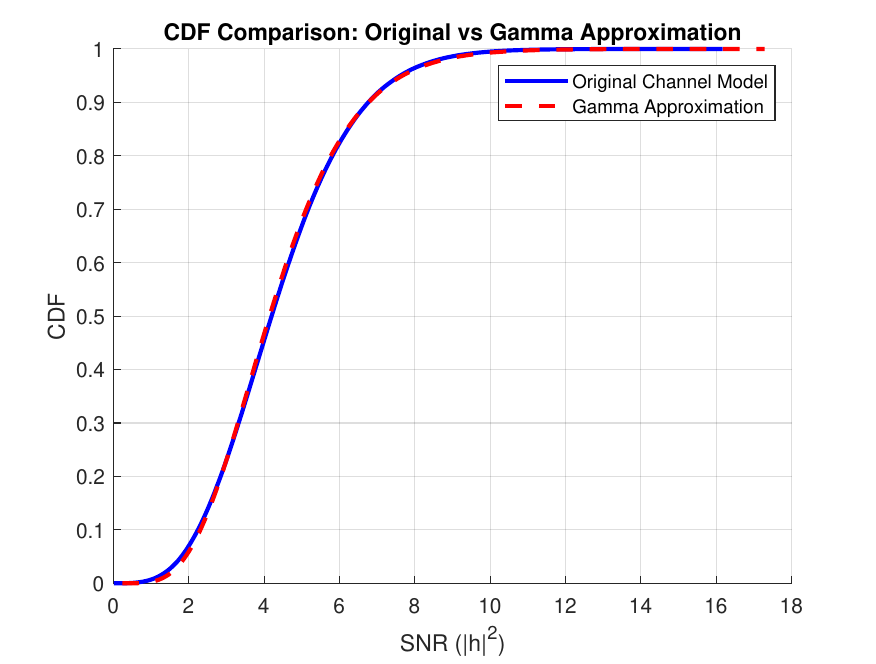}
	\vspace{-5mm}
	\caption{Simulation showing accuracy of Gamma approximation assuming no correlation}
	\vspace{-5mm}
	\label{fig:no_correlation}
\end{figure}

\vspace{-5mm}
\section{Derivation of the Second Moment considering correlation}
\label{appendix:correlation}
Let $|h_k|^2$ denote the signal power received at the $k$-th port in a fluid antenna system (FAS), where $k = 1, 2, \dots, N$. The second moment of the received power distribution is given by: $M_2 = \mathbb{E}\left[ |h_{\text{opt}}|^4 \right]$, where $h_{\text{opt}}$ is the channel coefficient at the port with maximum received power, and $M_2$ accounts for inter-port dependencies under spatial correlation.
%
\vspace{-5mm}
\subsection{Uncorrelated Case: Baseline Second Moment}
\vspace{-2mm}
For the uncorrelated case, the second moment {becomes}:%
\begin{align}\label{eqn:uncorrelated_m2}
	\tilde{M}_2 = \sum_{k=1}^{N} \tilde{M}_{2,k},
\end{align}
\noindent where $\tilde{M}_{2,k} = \mathbb{E}[ |h_k|^4 ]$ represents the uncorrelated second moment at the $k$-th port. When the ports have independent {fading}, the summation is sufficient. However, inter-port dependencies must be accounted for.
\subsection{Introducing Correlation via Covariance}
The correlation between the power levels at {the pair of} ports $i$ and $j$ is captured through their covariance:%
\begin{align}\label{eqn:covariance}
	\text{Cov}(|h_i|^2, |h_j|^2) = J_0\left(\frac{2\pi |i - j| W}{N-1}\right) \operatorname{Var}(|h_k|^2).
\end{align}

Using this, the second moment including correlation is:%
\begin{align}\label{eqn:m2_1}
	M_2 = \sum_{k=1}^{N} \tilde{M}_{2,k} + 2 \sum_{i \neq j} \operatorname{Cov}(|h_i|^2, |h_j|^2).
\end{align}

By substituting (\ref{eqn:covariance}) and (\ref{eqn:uncorrelated_m2}) into (\ref{eqn:m2_1}) {yields}:%
\begin{align}
	M_2 = \tilde{M}_2 + 2 \sum_{i \neq j} J_0\left(\frac{2\pi |i - j| W}{N-1}\right) \operatorname{Var}(|h_k|^2).
\end{align}
\vspace{-5mm}
\subsection{Avoiding Circular Dependence on $M_2$}
Since $\text{Var}(|h_k|^2) = M_2 - M_1^2,$ substituting this into the equation results in a recursive dependency on $M_2$, making it difficult to solve explicitly. Instead, a common approximation technique replaces $\text{Var}(|h_k|^2)$ with $M_1^2$, leading to: %
\begin{align}
	M_2 \approx \tilde{M}_2 + 2 \sum_{i \neq j} J_0\left(\frac{2\pi |i - j| W}{N-1}\right) M_1^2.
\end{align}

This simplification ensures that the second moment remains analytically tractable while still capturing correlation effects.
\vspace{-5mm}
\subsection{Final Expression for the Second Moment}

Thus, the final second moment of the correlated power distribution in a FAS is given by: %
\begin{align}\label{eqn:second_moment_correlated}
	M_2 = \tilde{M}_2 + 2 \sum_{i \neq j} J_0\left(\frac{2\pi |i - j| W}{N-1}\right) M_1^2,
\end{align}

\noindent where the second term introduces inter-port correlation using the Bessel function. Substituting (\ref{eqn:second_moment_correlated}) into (\ref{eqn:variance_formula}) and (\ref{eqn:gamma_parameter_def}), it can be seen in Figure \ref{fig:correlation_bad} that the Gamma approximation with inter-port correlation is highly inaccurate. This is {because} the correlation is overrepresented in the variance. To have a more accurate approximation, the correlation is adjusted and capped. 
\begin{figure}[!t]
	\centering
	\includegraphics[width=0.49\columnwidth]{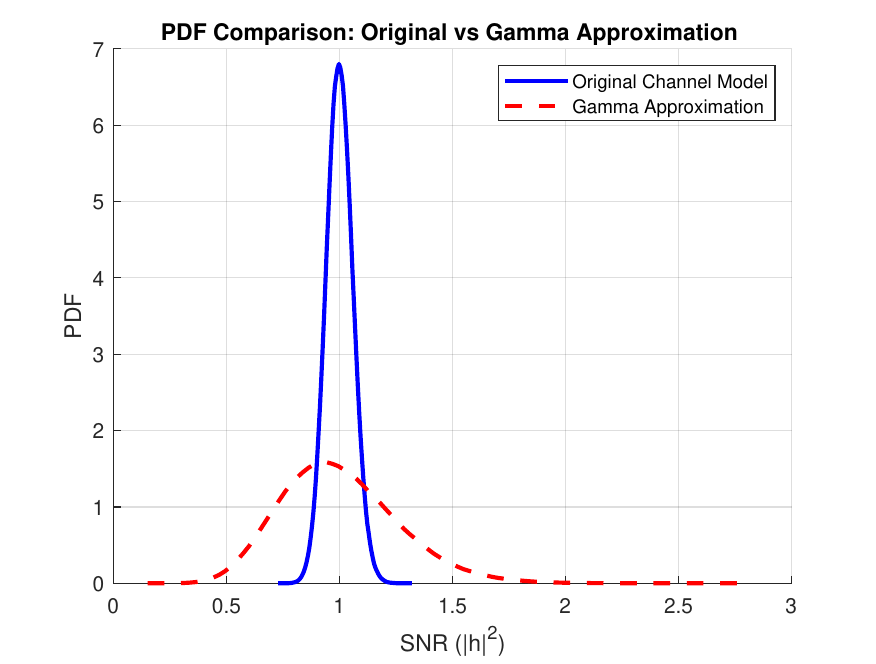}
	\hfill
	\includegraphics[width=0.49\columnwidth]{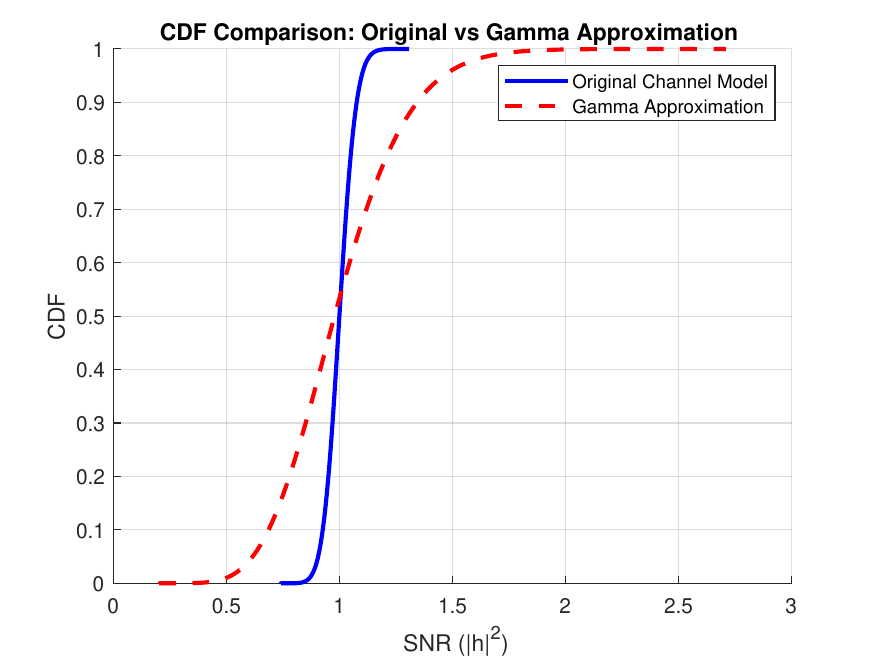}
	\vspace{-5mm}
	\caption{Simulation showing the accuracy of Gamma approximation {in the presence of} correlation.}
	\vspace{-5mm}
	\label{fig:correlation_bad}
\end{figure}
\vspace{-5mm}
\subsection{Adjusted Contribution}


As $N$ increases in (\ref{eqn:second_moment_correlated}),  the summation $ \sum_{i \neq j} $ grows quadratically since it contains $ N(N-1) $ terms. This causes the contribution to scale excessively as $ N $ increases, inflating the corrected variance. Upon dividing the summation by $ N^2 $, the growth of this quadratic term is controlled at large values of $N$, ensuring proportionality to $N$ {rather than to} $N^2$. The adjusted correlation $\bar{\text{Cov}}(|h_i|^2, |h_j|^2)$ {becomes}: %
\begin{align}\label{eqn:adjusted_correlation}
	\bar{\text{Cov}}(|h_i|^2, |h_j|^2) = \frac{1}{N^2} \sum_{i \neq j} J_0\left(\frac{2\pi |i - j| W}{N-1}\right) M_1^2.
\end{align}

With this adjustment, as $ N \to \infty $, the contribution approaches $ 2 M_1^2 $, preventing unbounded growth.


\vspace{-5mm}
\subsection{Capped Contribution}


Since normalization only has a significant impact on large $N$, the adjusted contribution can still dominate the corrected variance, when the ports are highly correlated and $ J_0(\cdot) $ is large, or if $ N $ is moderate, causing the $J_0 $-weighted summations to remain significant. Under these conditions, the corrected variance becomes unrealistically large, skewing the Gamma parameters ($ \alpha, \theta $) and inflating the Gamma approximation. The capped contribution ensures that the correlation effect remains proportional to the system's intrinsic variance: %
\begin{align}\label{eqn:capped_correlation}
	\hat{\text{Cov}}(|h_i|^2, |h_j|^2) = \min\left[\bar{\text{Cov}}, \eta \cdot \operatorname{Var}(|h_k|^2)\right],
\end{align}
where $\hat{\text{Cov}}(|h_i|^2, |h_j|^2)$ is the capped contribution and $\eta \in [0.1, 1]$ is a small constant. By limiting the correlation contribution to a fraction of the original variance, the correlation {will no longer} dominate the corrected variance. Capping avoids extreme {escalation} of $ \theta_{\text{gamma}} $, ensuring a reasonable fit between the Gamma model and the original channel.

To summarise, the adjusted contribution ensures that summation {becomes} proportional to $N$ not to $N^2$ for large numbers of ports, while the contribution cap prevents the exaggerated influence of the covariance term at high correlation.

By substituting (\ref{eqn:adjusted_correlation}) into (\ref{eqn:capped_correlation}), this can replace the previous correlation model in (\ref{eqn:second_moment_correlated}). Finally, this can be substituted into the original variance formula to obtain the variance considering the correlation in (\ref{eqn:variance}). Figure \ref{fig:correlation} verifies the accuracy of this updated correlation model.

\begin{figure}[!t]
	\centering
	\includegraphics[width=0.49\columnwidth]{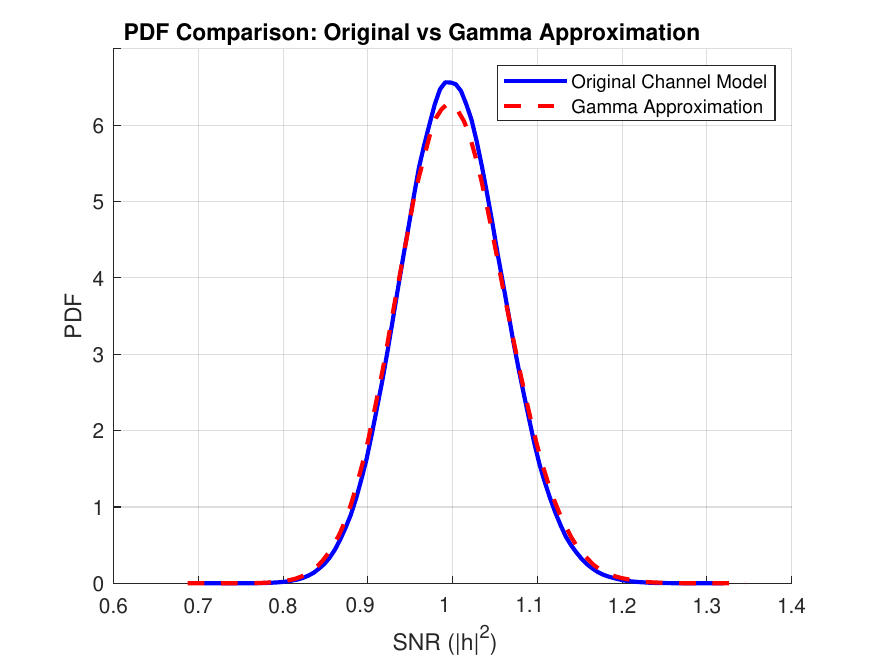}
	\hfill
	\includegraphics[width=0.49\columnwidth]{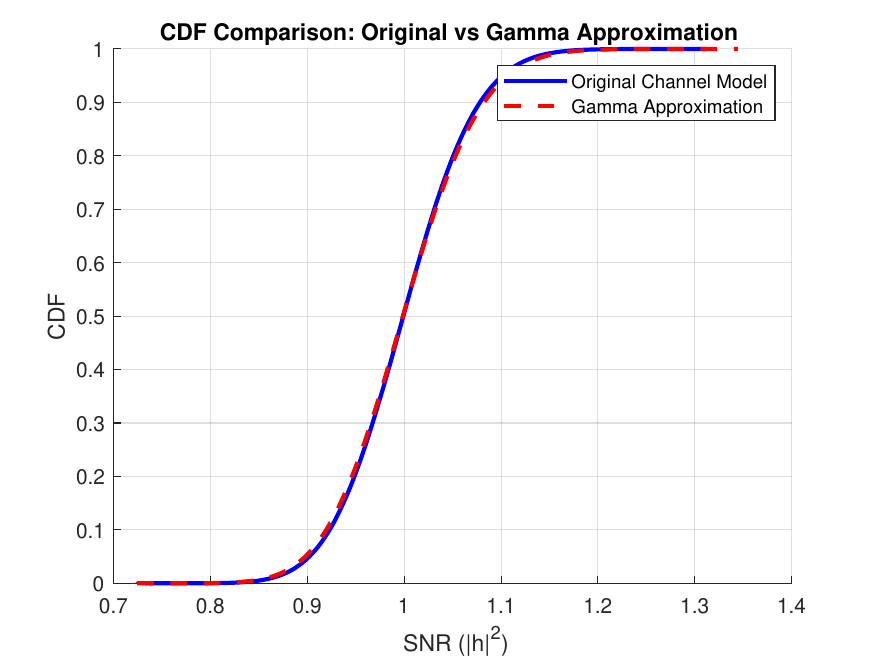}
	\vspace{-5mm}
	\caption{Simulation showing accuracy of Gamma approximation with the improved correlation model.}
	\vspace{-5mm}
	\label{fig:correlation}
\end{figure}



\vspace{-3mm}
\section{Derivation of exact outage expression}\label{appendix:exact_outage}

%
%
%
%

\subsection{Joint CDF Using the Correlation Matrix}

For correlated fading, the joint CDF is expressed using the determinant-based formulation {of}: %
\begin{align}\label{eqn:cdf_corr_matrix}
	P_{\text{out}}(\gamma_{\text{th}}) = \det \left( \mathbf{I} - \mathbf{R} \mathbf{F} \right),
\end{align}

\noindent where $\mathbf{R}$ is the spatial correlation matrix and $\mathbf{F}$ is a diagonal matrix with elements: $F_{|h_k|^2} (\gamma_{\text{th}}')$.

\vspace{-5mm}
\subsection{Laplace Expansion for Determinant-Based Outage}

The determinant can be expanded using Laplace expansion: %
\begin{align}\label{eqn:laplace}
	\det(\mathbf{A}) = \sum_{k=1}^{N} (-1)^{k+1} A_{1k} \det(\mathbf{A}_{1k}),
\end{align}

\noindent where $A_{1k}$ is the element in the first row and $k$-th column, {while} $\mathbf{A}_{1k}$ is the submatrix obtained by removing the first row and $k$-th column. Applying (\ref{eqn:laplace}) to (\ref{eqn:cdf_corr_matrix}), the final expression (\ref{eqn:outage_exact}) is obtained.





%
%

%
%


%
\vspace{-5mm}
\section{Outage Probability Bounds}\label{appendix:bounds}
\subsection{ Single-Port CDF \(F_S(\gamma)\)}
Assume $X_k = \lvert h_k \rvert^2$ follows a Gamma distribution with shape parameter $\alpha$ and scale parameter $\theta$. The CDF is: %
\begin{align}
F_S(\gamma) = \Pr\!\bigl(X_k < \gamma\bigr) = \frac{\gamma\bigl(\alpha, \tfrac{\gamma}{\theta}\bigr)}{\Gamma(\alpha)},
\end{align}
where $\gamma(\alpha, z)$ is the lower incomplete Gamma function and $\Gamma(\alpha)$ is the complete Gamma function.
	
\vspace{-5mm}
\subsection{Bounds on Outage Probability}

\subsubsection*{Lower Bound }

For fully correlated ports, the outage probability is the same as the single-port outage probability: %
\begin{align}
P_{\mathrm{out}}^{(\mathrm{lower})} = F_S(\gamma_{\mathrm{th}}) = \frac{\gamma\Bigl(\alpha,\;\tfrac{\gamma_{\mathrm{th}}}{\theta}\Bigr)}{\Gamma(\alpha)}.
\end{align}

\subsubsection*{Upper Bound (Fully Independent Ports)}

For fully independent ports, the outage probability is the \(N\)-th power of the single-port CDF: %
\vspace{-2mm}
\begin{align}
P_{\mathrm{out}}^{(\mathrm{upper})} = \bigl[F_S(\gamma_{\mathrm{th}})\bigr]^N = \left(
\frac{\gamma\Bigl(\alpha,\;\tfrac{\gamma_{\mathrm{th}}}{\theta}\Bigr)}{\Gamma(\alpha)}
\right)^N.
\end{align}
\vspace{-5mm}
\section{Exact Outage Probability of Single Path}\label{appendix:exact_outage_special}

By substituting the expressions for the pdf of $\lvert h_k \rvert$ into (\ref{eqn:outage_probability_def}), (\ref{eqn:exact_outage_v1}) is obtained. This can be further simplified to (\ref{eqn:exact_outage_v2}) using Fubini's theorem to swap the integration order. Finally, the inner integral can be further simplified using the Marcum-Q function to obtain (\ref{eqn:special_outage_probability_exact}).%
\begin{figure*}
	\vspace{-5mm}
    \begin{align}\label{eqn:exact_outage_v1}
		p(\gamma_{th}) =& P\!\left(\max_k |h_k| \leq \gamma_{\mathrm{th}} )\right) = \int^{\gamma_{th}}_0 \dots \int^{\gamma_{th}}_0 \frac{1}{2\pi} \times \\ \nonumber
		&\left[ \int^{\infty}_{\infty} \int^{\infty}_{\infty}	 \prod^N_{k=1} \left[ \frac{r}{\sigma^2}e^{- \frac{r^2 + |c_k|^2}{2\sigma^2}} \text{I}_0(\frac{r|c_k|}{\sigma^2})  \right]  \frac{1}{2\pi} e^{- \frac{x^2 + y^2}{2}} dx dy \right] dr_1 \dots dr_N.
	\end{align}
	\hrule
        \vspace{-4mm}
\end{figure*}
\vspace{-5mm}
\begin{figure*}
\vspace{-2mm}
\begin{align}\label{eqn:exact_outage_v2}
		p(\gamma_{th}) = \int^{\infty}_{\infty} \int^{\infty}_{\infty}	 \prod^N_{k=1} \left[ \int^{\gamma_{th}}_{0}  \frac{r}{\sigma^2}e^{- \frac{r^2 + |c_k|^2}{2\sigma^2}} \text{I}_0(\frac{r|c_k|}{\sigma^2}) dr_k \right]  \frac{1}{2\pi} e^{- \frac{x^2 + y^2}{2}} dx dy 
	\end{align}
	\hrule
    \vspace{-6mm}
\end{figure*}

\section{Applying the Gauss-Hermite Quadrature}\label{appendix:quadrature_outage}

For a Gaussian factor of $ e^{-t^2/2} $, the integration is adjusted in order to fit (\ref{eqn:gauss_hermite}) by using the substitution $ t = x / \sqrt{2} $: \vspace{-2mm}%
\begin{align}
	\int_{-\infty}^\infty e^{-x^2 / 2} g(x) \, \mathrm{d}x \approx \sqrt{2} \sum_{m=1}^M w_m g\left(\sqrt{2} x_m \right).
\end{align}

\vspace{-5mm}
\subsection{2D Integral for Outage Probability}

The outage probability integral can be rewritten as: %
\begin{align}
P_\text{out} = \frac{1}{2\pi} \iint_{-\infty}^\infty \exp\left(-\frac{x^2 + y^2}{2}\right) F(x, y) \, \mathrm{d}x \, \mathrm{d}y,
\end{align}
where $F(x, y) = \prod_{k=1}^N \left[1 - Q_1\left(\frac{|c_k(x, y)|}{\sigma}, \frac{\gamma_{\mathrm{th}}}{\sigma}\right)\right]$.

This can be rewritten s in terms of the Gauss-Hermite quadratures by substituting $ x = \sqrt{2} u $ and $ y = \sqrt{2} v $: %
\begin{align}
P_\text{out} = \frac{1}{\pi} \iint_{-\infty}^\infty e^{-(u^2 + v^2)} F(\sqrt{2} u, \sqrt{2} v) \, \mathrm{d}u \, \mathrm{d}v.
\end{align}

\vspace{-5mm}
\subsection*{Applying Guassian Quadrature to Outage Probability}

Using the product rule for 2D Gauss-Hermite quadratures, the double integral becomes: %
\begin{align}
&\iint_{-\infty}^\infty \exp\left[-(u^2 + v^2)\right] G(u, v) \, \mathrm{d}u \, \mathrm{d}v \\ \nonumber
&\approx \sum_{m=1}^M \sum_{n=1}^M w_m w_n G(x_m, x_n).
\end{align}

Substituting $G(u, v) = F(\sqrt{2} u, \sqrt{2} v) $ {yields}: %
\begin{align}
P(\text{outage}) \approx \frac{1}{\pi} \sum_{m=1}^M \sum_{n=1}^M w_m w_n F(\sqrt{2} x_m, \sqrt{2} x_n).
\end{align}

Expanding $F$, this becomes (\ref{eqn:quadrature_outage}).
\vspace{-5mm}
\section{Applying Gaussian Quadrature to Capacity}\label{appendix:quadrature_capacity}

The expression for capacity is approximated using the Gauss-Hermite quadrature for the outer integral $ \int_{-\infty}^\infty \int_{-\infty}^\infty$, and Gauss-Laguerre quadrature for the inner integral $ \int_0^\infty $.
\subsubsection*{ Reformulating the Outer Integral for Gauss-Hermite Quadrature}

The outer integral has a Gaussian weight $ \exp(-\frac{x^2 + y^2}{2}) $. Using the substitution $ x = \sqrt{2} u $, $ y = \sqrt{2} v $, and $ dx \, dy = 2 \, du \, dv $, the Gaussian weight becomes: \vspace{-3mm}%
\begin{align}
	\exp\left(-\frac{x^2 + y^2}{2}\right) = \exp\left(-u^2 - v^2\right).
\end{align}\vspace{-4mm}

Thus, the outer integral is rewritten as:\vspace{-2mm}%
\begin{align}
&\int_{-\infty}^\infty \int_{-\infty}^\infty \exp\left(-\frac{x^2 + y^2}{2}\right) G(x, y) \, dx \, dy \\ \nonumber 
&= 2 \int_{-\infty}^\infty \int_{-\infty}^\infty \exp\left(-u^2 - v^2\right) G(\sqrt{2} u, \sqrt{2} v) \, du \, dv.
\end{align}\vspace{-2mm}
This is suitable for the Gauss-Hermite quadrature:%
\begin{align}
&\int_{-\infty}^\infty \int_{-\infty}^\infty \exp\left(-u^2 - v^2\right) G(u, v) \, du \, dv \\ \nonumber 
&\approx \sum_{m=1}^M \sum_{n=1}^M w_m w_n G(\sqrt{2} x_m, \sqrt{2} x_n).
\end{align}\vspace{-2mm}

\subsubsection*{ Reformulating the Inner Integral for Gauss-Laguerre Quadrature}

The inner integral over $\gamma $ is:\vspace{-2mm}%
\begin{align}
	\int_0^\infty \frac{\prod_{k=1}^N Q_1\left( \frac{|c_k(x, y)|}{\sigma}, \frac{\sqrt{\gamma}}{\sigma} \right)}{1 + \gamma} \, d\gamma.
\end{align}\vspace{-5mm}

This can be reformulated with an exponential weight by {exploiting} $\frac{1}{1 + \gamma} = e^{-\ln(1 + \gamma)}$. Thus, the integral becomes:\vspace{-2mm}%
\begin{align}
\int_0^\infty \frac{\prod_{k=1}^N Q_1\left( \frac{|c_k(x, y)|}{\sigma}, \frac{\sqrt{\gamma}}{\sigma} \right)}{1 + \gamma} \, d\gamma = \int_0^\infty e^{-\gamma} H(\gamma) \, d\gamma,
\end{align}
where: \vspace{-2mm}%
\begin{align}
H(\gamma) = \prod_{k=1}^N Q_1\left( \frac{|c_k(x, y)|}{\sigma}, \frac{\sqrt{\gamma}}{\sigma} \right) e^{\gamma - \ln(1 + \gamma)}.
\end{align}\vspace{-2mm}

Using the Gauss-Laguerre quadrature {yields}:\vspace{-2mm}%
\begin{align}
\int_0^\infty e^{-\gamma} H(\gamma) \, d\gamma \approx \sum_{l=1}^L w_l H(t_l),
\end{align} \vspace{-2mm}

Combining the expressions for Gauss-Hermite and Gauss-Laguerre quadratures, (\ref{eqn:quadrature_capacity}) is obtained.


%
%

\bibliographystyle{IEEEtran}
\bibliography{reference}

\end{document}